\documentclass{sig-alternate}

\pdfoutput=1

\usepackage[usenames,dvipsnames]{color}
\usepackage[usenames,dvipsnames,svgnames,table]{xcolor}
\usepackage{graphicx}
\usepackage{subfigure}
\usepackage{verbatim}
\usepackage{hyperref}
\usepackage{url}
\usepackage{xspace}
\usepackage{wrapfig}
\usepackage{dsfont}
\usepackage{multirow}
\usepackage{caption}
\usepackage{etoolbox}

\newif\ifconfver
\newif\iffullver
\newif\ifdebug

\debugfalse
\confverfalse

\ifconfver
    \fullverfalse
\else
    \fullvertrue
\fi

\ifdebug
    \definecolor{mygreen}{rgb}{0.0, 0.5, 0.0}
    \definecolor{myorange}{rgb}{1.0, 0.49, 0.0}
    \definecolor{mypink}{rgb}{1.0, 0.01, 0.24}
    \definecolor{myblue}{rgb}{0.0, 0.0, 1.0}
\else
    \definecolor{mygreen}{rgb}{0.0, 0.0, 0.0}
    \definecolor{myorange}{rgb}{0.0, 0.0, 0.0}
    \definecolor{mypink}{rgb}{0.0, 0.0, 0.0}
    \definecolor{myblue}{rgb}{0.0, 0.0, 0.0}
\fi

\newenvironment{changed}[0]{\color{myblue}}{} 

\ifconfver
    \newenvironment{fullver}[0]{\color{mygreen}\expandafter\comment}{\expandafter\endcomment}
    \newenvironment{confver}[0]{\color{myorange}}{}
\else
    \newenvironment{fullver}[0]{\color{mygreen}}{}
    \newenvironment{confver}[0]{\color{myorange}\expandafter\comment}{\expandafter\endcomment}
\fi

\iffullver
    \makeatletter
    \def\@copyrightspace{\relax}
    \makeatother
\fi

\newcommand{\onebb}{\mathds{1}} 
\newcommand{\calh}{\mathcal{H}}
\newcommand{\fhat}{\hat{f}}
\newcommand{\newfhat}{f^\dagger}
\newcommand{\bigoh}[1]{\mathcal{O}(#1)}

\newcommand{\tuple}[1]{\ensuremath{\left \langle #1 \right \rangle }}
\DeclareMathOperator*{\argmin}{arg\,min}

\newcommand{\sentpos}[0]{``$+$''\xspace}
\newcommand{\sentneg}[0]{``$-$''\xspace}

\newcommand{\memd}{\text{EMD}}
\newcommand{\emd}{$\memd$\xspace}

\newcommand{\memdbanks}{\text{EMD}^{\alpha}}
\newcommand{\emdbanks}{$\memdbanks$\xspace}

\newcommand{\memdpelwer}{\widehat{\text{EMD}}}
\newcommand{\emdpelwer}{$\memdpelwer$\xspace}

\newcommand{\memdnew}{\text{EMD}^{\star}}
\newcommand{\emdnew}{$\memdnew$\xspace}

\newcommand{\memdsent}{\text{SND}}
\newcommand{\emdsent}{$\memdsent$\xspace}

\newcommand{\dex}{\widetilde{D}}
\newcommand{\dexpp}{\widetilde{D(P, +)}}
\newcommand{\pex}{\widetilde{P}}
\newcommand{\pexp}{\widetilde{P^+}}
\newcommand{\qex}{\widetilde{Q}}
\newcommand{\qexp}{\widetilde{Q^+}}

\newtheorem{theorem}{Theorem}

\newtheorem{corollary}{Corollary}
\newtheorem{lemma}{Lemma}

\newcommand{\vsa}{\vspace*{-0.1cm}}
\newcommand{\vsb}{\vspace*{-0.2cm}}
\newcommand{\vsbb}{\vspace*{-0.3cm}}
\newcommand{\vsc}{\vspace*{-0.4cm}}

\begin{document}


\conferenceinfo{WWW}{'16 Montr{\'e}al, Canada}

\title{A Distance Measure for the Analysis of \\
	Polar Opinion Dynamics in Social Networks}
	

\numberofauthors{3}
\author{
    \alignauthor
    Victor Amelkin\\
    \affaddr{University of California}\\
    \iffullver
        \affaddr{Santa Barbara, CA}\\
    \fi
    \email{victor@cs.ucsb.edu}
    \alignauthor
    Ambuj K. Singh\\
    \affaddr{University of California}\\
    \iffullver
        \affaddr{Santa Barbara, CA}\\
    \fi
    \email{ambuj@cs.ucsb.edu}
    \alignauthor
    Petko Bogdanov\\
    \affaddr{University at Albany -- SUNY}\\
    \iffullver
        \affaddr{Albany, NY}\\
    \fi
    \email{pbogdanov@albany.edu}
}

\maketitle

\vspace{-0.1in}

\begin{abstract}
Analysis of opinion dynamics in social networks plays an important role
in today's life. For applications such as predicting users' political
preference, it is particularly important to be able to analyze the dynamics
of competing opinions. \emph{While observing the evolution of polar
opinions of a social network's users over time, can we tell when the
network ``behaved'' abnormally? Furthermore, can we predict how the opinions of the users will change in the future? Do opinions evolve according to existing network opinion dynamics models?} To answer such questions, it is not sufficient to study individual user behavior, since opinions can spread far beyond users' egonets. We need a
method to analyze opinion dynamics of all network users simultaneously
and capture the effect of individuals' behavior on the global evolution pattern
of the social network.

In this work, we introduce Social Network Distance (\emdsent)
---a distance measure that quantifies the ``cost''
of evolution of one snapshot of a social network into another snapshot under
various models of polar opinion propagation.
\emdsent has a rich
semantics of a transportation problem, yet, is computable in time linear in the number of users, which makes \emdsent applicable
to the analysis of large-scale online social networks. In our experiments
with synthetic and real-world Twitter data, we demonstrate the utility of our distance measure for
anomalous event detection. It achieves a true positive rate of $0.83$,
twice as high as that of alternatives.
When employed for opinion prediction in Twitter,
our method's accuracy is $75.63\%$, which is $7.5\%$ higher
than that of the next best method.

\vspace{0.05in}
\noindent{\textbf{Code:}} \url{http://cs.ucsb.edu/~victor/pub/ucsb/dbl/snd/}
\end{abstract}

\vsc
\section{Introduction}

Analysis of people's opinions plays an important role in today's life, and
social networks provide a great platform for such research. Businesses are
interested in advertising their products in social networks relying on
viral marketing. Political strategists are interested in predicting an
election outcome based on the observed sentiment change of a sample of voters.
Mass media and security analysts may be interested in a timely discovery of
anomalies based on how a social network ``behaves''. Thus, it is important
to have methods for the analysis of how user opinions evolve in a social
network. 

\emph{How can we quantify the change in opinions of users with respect to
their expected behavior in a social network? Can we distinguish opinion
shifts caused by in-network user interaction from those caused by
factors external to the network?}
To answer such questions, we need a distance
measure that explicitly models opinion dynamics, incorporating both the
distribution of user opinions at two time instances and the
network structure that defines the pathways for opinion dissemination.
In this work, we develop such a distance measure for
snapshots of a social network and employ it for the
analysis of competing opinion dynamics.

While the dynamics of a social network can be characterized
by evolution of both the network's structure and user opinions~\cite{myers2014bursty},
in this paper we focus on the opinion dynamics.
We assume that there are two \emph{polar opinions} in the network,
\emph{positive} \sentpos and \emph{negative} \sentneg. Users having no
or an unknown opinion are termed \emph{neutral}, while those expressing
opinion---\emph{active}. A \emph{network state} is comprised of the opinions of all
network users at a given time. Polar opinions \emph{compete} in
that a user having a positive opinion is unlikely to enthusiastically
spread information about the adverse negative opinion, yet, would spread
information about the friendly positive opinion ``at a discount cost''.
Such competition may arise when the notions the opinions relate to are inherently
competing. For example, in an election, the voters leaning toward one political
party are unlikely to spread positive rumors about the competing party. Another
example is viral marketing, where the consumers who favor smartphones of one
brand may readily express their affection to it, but may be unwilling to praise
the competing brand.

Given a time series of network states, we address the applications of detecting
anomalous network states and predicting opinions of individual users. For the first application,
we answer the question of which network states are anomalous with respect to the
observed evolution of the network. For the second application, we predict currently
unknown opinions of selected users in the network based on the network's observed
past and present behavior.

The analysis of a time series of network states is, however, complicated,
because network states do not naturally belong to any vector space, and the
existing time series analysis techniques cannot be readily applied. Our
approach is to treat network states as members of a metric space induced
by a distance measure governed by both the network's structure and user
opinions. We design a semantically and mathematically appealing as well as
efficiently computable distance measure \emph{Social Network Distance (\emdsent)}
for the comparison of social network states containing polar opinions, and
demonstrate its applicability to the analysis of real-world data.

To quantify the dissimilarity between network states, \emdsent takes into account how information propagates in the network. A change of a given user's opinion from, say, neutral to positive, contributes to the overall distance between the corresponding network states by reflecting the likelihood of this user's opinion change based on the opinions and locations of other users in the network under a chosen model of polar opinion dynamics. However, since the network users interact, the distance measure needs to consider the opinion shifts of all users simultaneously. Thus, we define \emdsent as a transportation problem that models opinion spread and adoption in the network. In particular, by making the transportation costs dependent on both the network's topology and the opinions of the users conducting information in the network, we capture the competitive aspect of polar opinion propagation.

The summary of \emph{our contributions} is as follows:\\
\noindent $\bullet$ We propose \emdsent---the first distance measure suitable for comparison of social network states containing competing opinions under various models of opinion dynamics.\\
\noindent$\bullet$ We develop a scalable method to precisely compute \emdsent in time linear in
        the number of users in the network, thus, making \emdsent applicable to the
        analysis of real-world online social networks.\\ 
\noindent$\bullet$ We demonstrate the applicability of our distance measure using
        both synthetic and real-world data from Twitter. In detecting anomalous
        states of a social network, \emdsent is superior to other distance
        measures
        in discovering controversial events that have polarized the society.
        In user opinion prediction experiments, \emdsent also outperforms competitors
        in terms of prediction accuracy.

\vsbb
\section{Earth Mover's Distance
and\\ Network State Comparison}

A good distance measure for the states of a social network should take
into account the specifics of polar opinion
propagation in a network. For example, a user having opinion \sentpos should not
actively participate in the dissemination of opinion \sentneg, or, at least, the
dissemination of an adverse opinion should incur a large cost. On the other hand,
this user should disseminate friendly opinion \sentpos at a cost lower than the
cost a neutral user would incur. Thus, we propose to address the problem of
comparing states of a social network as a transportation problem where the
costs of opinion propagation are defined based on the shortest paths between
the users in the network, computed taking into account both
the network structure and the user opinions.

This high-level intuition about network state comparison as
an opinion transportation problem inadvertently leads us to one of the well-studied
distance measures---Earth Mover's Distance (EMD). Originally, defined as a dissimilarity
measure for histograms~\cite{emd-rubner}, EMD can be used for the comparison
of network states viewed as histograms, with histogram bins' values quantifying individual user opinions.
Intuitively, EMD measures the costs of optimal transformation
of one histogram into another with respect to the \emph{ground distance} specifying
the costs of moving mass between bins. In our case, the ground distance
is defined based on the shortest paths between the users of the network, where
the shortest paths depend on both the network's topology as well as the opinions of
the users facilitating opinion propagation.

Formally, given two real-valued histograms $P = [P_1, \dots, P_n]$
and $Q = [Q_1, \dots, Q_m]$, EMD between them with respect to a cross-bin
ground distance $\{D_{ij}\}_{n \times m}$ is defined as the solution to the
problem of optimal mass transportation from the bins of $P$ (suppliers) to
the bins of $Q$ (consumers) with respect to transportation costs $D$.
\vsa
\begin{align}
	\memd(P, Q, D) =
		\sum\limits_{i=1}^{n}{\sum\limits_{j=1}^{m}{
		    D_{ij} \widehat{f}_{ij}
		}} ~\big/~ 
		\sum\limits_{i=1}^{n}{\sum\limits_{j=1}^{m}{
		    \widehat{f}_{ij}
		}},
\label{eq:emd}
\end{align}
\vsb

\noindent where $\{\widehat{f}_{ij}\}_{n \times m}$ is an optimal transportation plan in the following transportation problem:
\begin{align*}
	&\sum\limits_{i=1}^{n}{\sum\limits_{j=1}^{m}{
			f_{ij} D_{ij}
	}} \to \min,~
	\sum\limits_{i=1}^{n}{\sum\limits_{j=1}^{m}{
			f_{ij}
	}} = \min\left\{
    	    \sum\limits_{i=1}^{n}{P_i},
    	    \sum\limits_{j=1}^{m}{Q_j}
	    \right\},\\
	&f_{ij} \geq 0,~
    \sum\limits_{j=1}^{m}{ 
		f_{ij}
	} \leq P_i,~
    \sum\limits_{i=1}^{n}{
		f_{ij}
	} \leq Q_j,~
    (1 \leq i \leq n, 1 \leq j \leq m).
\end{align*}

\emd is attractive not only because its ground distance can capture how opinions
propagate in the underlying network, but it is also driven by node states rather
than network topology, is spatially-sensitive, and metric under
\ifconfver
certain conditions~\cite{emd-rubner}.
\else
the following conditions.
\fi
\iffullver
\begin{fullver}
    \begin{theorem}{\emph{(Rubner et al.~\cite{emd-rubner})}}
        If all histograms under comparison have equal total masses, and the underlying ground distance
        is metric, then \emd is metric.
        \label{thm:emd-metric}
    \end{theorem}
\end{fullver}
\fi
In the following section, we use \emd to construct
a distance measure for network states containing polar opinions.

\vsbb
\section{Distance Measure for Network\\ States with Polar Opinions}
Given a network $G = \tuple{V, E}$, where $V$ $(|V| = n)$ is the set of nodes (users)
and $E$ is the set of edges (social ties), we want to compute the distance between two of
 its states $P = [P_1, \dots, P_n]^T$ and $Q = [Q_1, \dots, Q_n]^T$. While generalizations
 are possible, we use
 an intuitive and simple scheme for polar opinion quantification: if user $i$ has opinion \sentpos in network state $P$, then $P_i = +1$;
$P_i = -1$ if the user's opinion is \sentneg; and $P_i = 0$ if the user is neutral\footnote{
    \begin{changed}
    There is a great body of research on methods for opinion classification based on
    user-generated content. In this work, however, we
    assume that we have access only to the quantified opinions, and no access to the
    user-generated content (e.g., tweets), which may be unavailable due to privacy
    reasons.
    \end{changed}
}.

Despite the appeal of \emd, it is not readily applicable to the comparison of
network states $P$ and $Q$, since
(i) the users' behavior may change in the process of opinion propagation,
while a transportation problem underlying \emd operates with static
transportation costs; (ii) \emd is defined for histograms of a
homogeneous quantity, while $P$ and $Q$ contain both positive and negative values; and (iii) it is not clear how
to incorporate the node states into the definition of the ground distance. In
order to define our Social Network Distance (\emdsent), we address these three
problems in what follows.

\vspace{0.01in}
\noindent\emph{\bf (i) \emdsent as a transportation problem.}
For two given users $u$ and $v$, not necessarily being immediate neighbors in
the network, the cost of $v$'s adopting opinion from $u$ depends not only on
the states and locations of $u$ and $v$ in the network, but also on the states
of the users through which $u$'s opinion can reach $v$.
\ifconfver
    \vspace{-0.12in}
\fi
\begin{figure}[h!]
    \centering
	\includegraphics[width=2.8in]{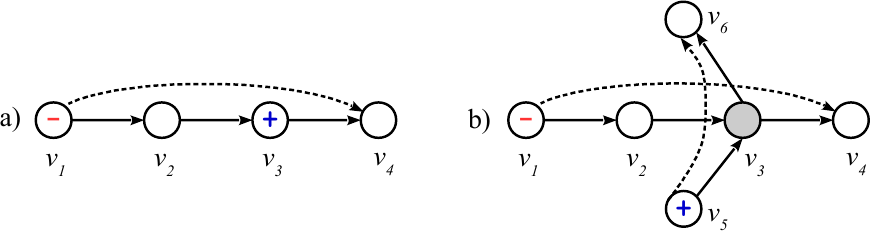}
	\vsb
	\caption{Transitive opinion propagation in a social network. The solid and dashed arrows represent social ties and opinion flow, respectively.}
	\label{fig:social-net-trans-learn}
\end{figure}
\ifconfver
    \vspace{-0.15in}
\else
    
\fi

In Fig.~\ref{fig:social-net-trans-learn}.a, user $v_3$ having opinion
\sentpos affects the cost of propagating opinion \sentneg from user $v_1$
to user $v_4$. In Fig.~\ref{fig:social-net-trans-learn}.b, however, user
$v_3$ is initially neutral, but can become active in the process
of $v_5$'s propagating \sentpos to $v_6$ before opinion \sentneg from $v_1$
reaches $v_3$, thereby, impeding the spread of \sentneg from $v_1$ to $v_4$.
In order to pose \emdsent as a transportation problem, we will
assume that the costs of opinion propagation depend only on the opinions of
the currently active users, taking no account of the potential change of
user opinions in the process of opinion transportation.

\vspace{0.10cm}
\noindent\emph{\bf (ii) Handling polar opinions.} We design \emdsent to measure the optimal cost
of transforming one network state $G_1$ into another network state $G_2$ by the
means of opinion transportation. The active users of $G_1$ are the suppliers
and those in $G_2$ are the consumers in the transportation problem
setting. In defining the transportation costs, we assume that users adopting, say, opinion
\sentpos, in $G_2$ are affected by others of the same opinion during the propagation. 
Similarly, suppliers only propagate opinions to the consumers of the same type.
Thus, the set of constraints in the transportation problem can be divided into
two non-overlapping subsets for two kinds of opinions. Consequently, the
problem of optimal transportation of opinions from suppliers of $G_1$ to
consumers of $G_2$ can be split into two transportation problems:
one for transporting the opinions of each kind.

\vspace{0.10cm}
\noindent\emph{\bf(iii) Defining the ground distance.}
The cost of opinion propagation from user $u$ to user $v$ depends on their
topological proximity, \begin{changed}how frequently they communicate\end{changed}, persuasiveness, and
stubbornness of $u$ and $v$ as well as the users ``separating'' them.
Formally, the ground distance $D(G_i, op) \in {\mathbb{R}^+}^{n \times n}$,
reflecting the costs of propagating opinion $op$ through a network in state $G_i$,
is a matrix containing the lengths of the shortest paths in a network with its adjacency
matrix defined as:

\vsc
\begin{align}
	A^{ext}&(G_i, op) = \notag \\
	-&\log \mathbb{P}(G_i, op) - \log \mathbb{P}^{in}(G_i, op) - \log \mathbb{P}^{out}(G_i, op),
    \label{eq:emdsent-extended-adjmat}
\end{align}
\vsc

\noindent where the summands on the right are $n$-by-$n$ matrices of log-probabilities
of communication,
opinion adoption, and opinion spreading, respectively. Probabilities $\mathbb{P}(G_i, op)$
can be defined as the relative frequencies of communication between users.
In the absence of such information, we set $-\log{\mathbb{P}(G_i, op)}$ to be the
connectivity matrix of the network, penalizing for the users' topological remoteness.
Opinion adoption probabilities $\mathbb{P}^{in}(G_i, op)$
reflect users' susceptibility/stubbornness~\cite{yildiz2011discrete}. If such
information is unavailable, for each existing edge $\tuple{u,v}$, we set $\mathbb{P}^{in}_{uv} = 1$,
\begin{changed}so that all users are non-stubborn and equally receptive to persuasion\end{changed}.
Finally, the opinion spreading penalties $-\log \mathbb{P}^{out}(G_i, op)$ are defined
based on a particular opinion dynamics model. Several ways to make such an assignment
are described below.

\begin{changed}
\emph{Model-agnostic Opinion Propagation:}
If there is no evidence that opinions evolve in the network according to a specific
opinion dynamics model, then the opinion spreading penalties can be defined as\end{changed}

\vsc
{\small
\begin{align*}
	-\log \mathbb{P}^{out}_{uv}(G_i, op) = \begin{cases}
		c_{adverse} & \text{if } G_i[u] \neq op \vee G_i[v] = -op,\\
		c_{neutral} & \text{if } G_i[u] = 0,\\
		c_{friendly} & \text{if } G_i[u] = op,
	\end{cases}
\end{align*}
}
\vsbb

\noindent where $c_{adverse}$, $c_{neutral}$, $c_{friendly}$ are constant penalties for
spreading opinion $op$ by the users having respectively adverse, neutral, or friendly opinion
relative to $op$\begin{changed}, and $G_i[u]$ is the opinion of user $u$ in network state $G_i$\end{changed}.
This simple \begin{changed}definition implies\end{changed} that users willingly
spread opinions similar to their own ($c_{friendly}$ is small); are unwilling to
spread adverse opinions ($c_{adverse}$ is large); with neutral users' behavior
being somewhere in-between ($c_{friendly} < c_{neutral} < c_{adverse}$).

\begin{changed} Alternatively, $\mathbb{P}^{out}(G_i, op)$ can be defined via one of the existing
opinion dynamics models discussed next.\end{changed}

\vspace{0.02in}
\emph{Independent Cascade Model:} The distance-based
model of Carnes et al.~\cite{carnes2007maximizing} is a version of Independent Cascade Model
capturing opinion competition. In this model, we have two sets of initial
adopters $I_+$, $I_-$ of opinions \sentpos and \sentneg, respectively, with
$I = I_+ \cup I_-$. Each edge $\tuple{u,v}$ is labeled with an activation probability
$p_{uv}$ (which can be learned from the observed data~\cite{goyal2010learning})
and a distance $d_{uv}$. If we denote by $d_v(I)$ the shortest distance from any user
of set $I$ to user $v$ with respect to edge distances $d_{uv}$, and denote by
$p^{a}(G_i, v)$ the sum of edge activation probabilities $p_{uv}$ taken over
all active users $u$ in $G_i$ such that $d_v({u}) = d_v(I)$, then

\vsbb
{\small
\begin{align*}
	\mathbb{P}^{out}_{uv}(G_i, op) = \begin{cases}
	    0 & \text{if } d_v(\{u\}) > d_v(I),\\
	    1 & \text{if } G_i[u] = op \wedge G_i[v] = op,\\
	    \frac{\max(0, p_{uv} - \epsilon)}{p^{a}(G_i, v)} & \text{if } G_i[u] = op \wedge G_i[v] = 0,\\
	    \epsilon & \text{otherwise}.
	\end{cases}
\end{align*}
}
\vsb

In \begin{changed}the original model of\end{changed}~\cite{carnes2007maximizing}, $\epsilon = 0$, that is, neutral users
cannot infect others, and active users do not drop their opinions or spread
competing opinions. If, however, we compare network states with respect to
the original model, then many network states derived from real-world data would be
at distance $+\infty$ from each other, either due to the lack of knowledge
about the network (an edge has not been observed or a user's opinion has
been \begin{changed}misclassified\end{changed}) or due to an imperfect fit of the model and the data.
Instead of just declaring two network states as qualitatively unreachable,
we always want to quantify the distance between them, and, thus, assign
some negligible probabilities $\epsilon$ to the events that the model posits
as impossible.

\vspace{0.02in}
\emph{Linear Threshold Model:} A version of Linear
Threshold Model supporting opinion competition has been proposed by Borodin
et al.~\cite{borodin2010threshold}. In this model, each edge $\tuple{u,v}$
is weighted with $\omega_{uv}$ reflecting the amount of influence $u$ has
over $v$;
and each user $u$ has an in-advance chosen constant threshold $\theta_u$.
If we denote by $N^{in}(G_i, v)$ the set of in-neighbors of $v$ active in
$G_i$, and by $\Omega^{in}$ the sum of $\omega_{xv}$ over all $x \in N^{in}(G_i, v)$,
then

\vsc\vsa
{\small
\begin{align*}
	\mathbb{P}^{out}_{uv}(G_i, op) = \begin{cases}
	    0 & {\footnotesize \text{if } u \notin N^{in}(G_i, v),}\\
	    1 & {\footnotesize \text{if } G_i[u] = op \wedge G_i[v] = op,}\\
	    \frac{(1-\epsilon)\omega_{uv}}{\Omega^{in}} &
            {\footnotesize
                \text{if } G_i[u] = op \wedge G_i[v] = 0 \wedge \Omega^{in} \ge \theta_v,
            }\\
	    \epsilon & {\footnotesize otherwise}.
	\end{cases}
\end{align*}
}
\vsb

Having addressed challenges (i)-(iii), we are now ready to formally define \emdsent.

In general, the network's structure might have changed between the times
corresponding to the two network states under comparison~\cite{myers2014bursty},
but defining the ground distance for each pair of users over a different network
would incur an unacceptably high time complexity of the resulting distance
measure.
Thus, for time-ordered network states $G_1$ and $G_2$, one can define the
ground distance based on the network structure corresponding to the earlier network
state $G_1$.
However, to make \emdsent applicable to time-unordered network states
as well, we define \emdsent based on both $D(G_1, op)$ and $D(G_2, op)$ as follows.

\vsc
{\small
\begin{align}
\label{eq:emdsent-def}
    &\memdsent(G_1, G_2) = \frac{1}{2} \times \large[ \\
    &\memd(G_1^+, G_2^+, D(G_1, +)) + \memd(G_1^-, G_2^-, D(G_1, -)) + \nonumber \\
    &\memd(G_2^+, G_1^+, D(G_2, +)) + \memd(G_2^-, G_1^-, D(G_2, -)) \large], \nonumber
\end{align}
}
\vsc

\noindent where users having opinion \sentneg are considered neutral in $G_i^+$, users
having opinion \sentpos are neutral in $G_i^-$, and \emd is the original
Earth Mover's Distance~\cite{emd-rubner} that will further be replaced with
its generalization \emdnew designed in the next section. Since \emdsent is a linear
combination of several instances of \emd, then \emdsent preserves most mathematical
properties, and, in particular, metricity of the chosen \emd.

\vsb
\section{Generalized Earth Mover's\\ Distance}
The original \emd~\cite{emd-rubner} is limited in that it cannot adequately
compare histograms having different total mass---it ignores the mass
mismatch, assigning a small distance value to a pair of a light and a heavy
histograms. However, if we think about two
histograms corresponding to the states of a social network, one with a few
and another one with many active users, then we expect the distance between
such histograms to be large. In real-world data, even consecutive network states
observe a widely varying number of active users, making the
challenge of comparing histograms with total mass mismatch well pronounced.

There are several extensions of \emd that address the above mentioned limitation.
One of them, \emdpelwer~\cite{emd-pelwer}, augments \emd with an additive mass
mismatch penalty as follows
%
{\small
\begin{align*}
    &\memdpelwer(P, Q, D) =
		\ \memd(P, Q, D) \cdot
		\min{\big\{
			\sum{P_i}, \sum{Q_j}
		\big\}} + \nonumber \\
		&\phantom{xxx}+\ \alpha \cdot \max{\{D_{ij}\}} \cdot \big|
		\sum{P_i} - \sum{Q_j}
	\big|,
\end{align*}
}
where \emd is the original Earth Mover's Distance, and $\alpha$ is a constant
parameter. \begin{changed}The second summand represents the mass mismatch penalty that 
depends only\end{changed}
on the magnitude of the mass mismatch and \begin{changed}the maximum ground distance,
thereby, being unable to capture the fine details of the network's structure $D$ can
depend upon\end{changed}. This is, however,
inadequate for the comparison of the states of a social network, because the network's
behavior depends not only on the number of new activations, but also on where these
newly activated users are located in the network.

Another \emd version, namely, \emdbanks~\cite{emd-banks}, extends each histogram with an
extra bin (\emph{``the bank bin''}) whose value is chosen so that the total masses of the histograms become
equal.
\begin{fullver}
An example of such an extension is shown in Fig.~\ref{fig:emd-banks}.
\begin{figure}[h!]
    \centering
	\includegraphics[width=3in]{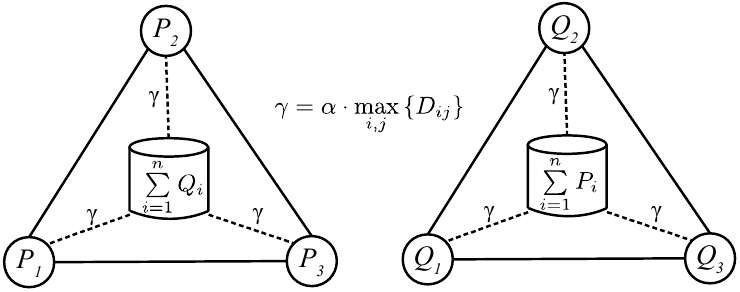}
	\caption{
        Histograms $P$ and $Q$ defined over the same network are extended with bank bins, whose capacities are chosen, so that the total masses of the
        extended histograms $\pex$ and $\qex$ are equal. The ground distances
        $\dex_{bank,i} = \dex_{i, bank} = \gamma$ from and to the bank bin are
        uniformly defined based on the largest ground distance between the
        initially present bins.
	}
	\label{fig:emd-banks}
\end{figure}
Formally, \emdbanks is defined as follows.
\begin{align*}
    & P = [P_1, \dots, P_n],\ Q = [Q_1, \dots, Q_n],\\
    & P_{bank} = \sum\limits_{j=1}^{n}{Q_j}, \hspace{0.1in}
        \widetilde{P} = \left[ P, P_{n+1} = P_{bank} \right], \notag \\
    & Q_{bank} = \sum\limits_{i=1}^{n}{P_i}, \hspace{0.1in}
		\widetilde{Q} = \left[ Q, Q_{n+1} = Q_{bank} \right], \\
	\dex = &\left[ \begin{array}{c|c}
	    D_{n \times n} &
	    \begin{array}{c} |\\ \alpha \max\limits_{i,j}{\{D_{ij}\}}\\ | \end{array} \\
	    \hline
	    \mbox{---}~\alpha \max\limits_{i,j}{\{D_{ij}\}}~\mbox{---} & 0
    \end{array} \right],
\end{align*}
\begin{align*}
	\memdbanks&(P, Q) =
	\memd(\widetilde{P}, \widetilde{Q}, \dex) \cdot
		\left(
			\sum\limits_{i=1}^{n}{P_i} +
			\sum\limits_{j=1}^{n}{Q_j}
		\right).
\end{align*}
\end{fullver}
However, as we establish in Theorem~\ref{thm:emd-banks-pelwer-equiv}, \emdbanks is equivalent
to \emdpelwer and, hence, is also unsuitable for the comparison of social network states.
\vsa
\begin{theorem}
    If ground distance $D \in \mathbb{R}^{n \times n}$ and parameter $\alpha \in \mathbb{R}^+$
    are chosen such that both \emdbanks and \emdpelwer are metric, that is, $D$ is metric
    and $\alpha \geq 0.5$~\cite{emd-banks, emd-pelwer}, then
    $\forall P, Q \in \mathbb{R}^n: \memdbanks(P, Q, D) = \memdpelwer(P, Q, D)$.
\label{thm:emd-banks-pelwer-equiv}
\end{theorem}
\begin{confver}
The proof of Theorem~\ref{thm:emd-banks-pelwer-equiv} is provided in the extended
version of this paper~\cite{snd-full}.
\end{confver}
\begin{fullver}
\begin{proof}
    W.l.o.g., let us assume that $\sum{P_i} \leq \sum{Q_j}$.
    The proof will use the following notation:
	$$
		\Delta = \Delta(P,Q) = \left| \sum\limits_{i=1}^{n}{P_i} - \sum\limits_{j=1}^{n}{Q_j} \right|, \hspace{0.1in}
		\gamma = \alpha \max\limits_{i,j}{\{D_{ij}\}}.
	$$
	Hence, \emdpelwer is defined as
	\begin{align*}
	    \memdpelwer(P, Q) =
    	    \memd(P, Q) \min{\left\{
    			\sum\limits_{i=1}^{n}{P_i}, \sum\limits_{j=1}^{n}{Q_j}
    		\right\}} +
    		\gamma \Delta.
    \end{align*}
    The goal of the proof is to show that \emdbanks has exactly the same
    expression as \emdpelwer, as long as they are metric.
    Consider how a
    unit of mass can be transported between two histograms (see Fig.~\ref{fig:emd-eq-proof}).
\begin{figure}[h!]
    \centering
	\includegraphics[width=3in]{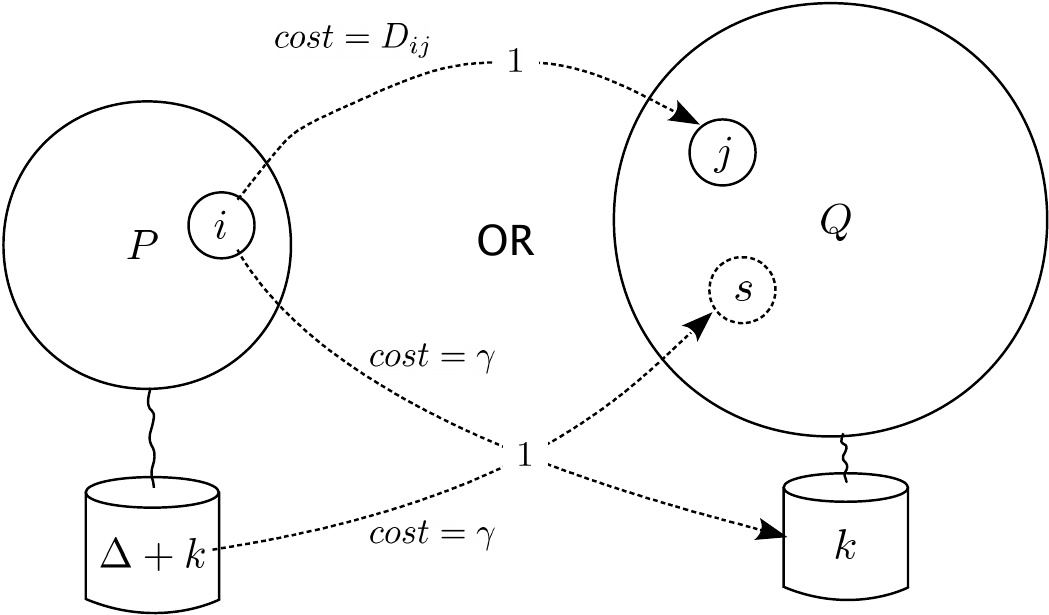}
	\caption{
	    Two qualitatively different ways to transport a unit of mass from extended
	    histogram $\pex = [P, k + \Delta]$ to extended histogram $\qex = [Q, k]$. (Dashed arrows represent the
	    flow of mass.) The bank bin is attached to every node of each histogram.
	    $k = \sum{P_i}$, so that the total masses of two histograms are equal.
	}
	\label{fig:emd-eq-proof}
\end{figure}
There are two qualitatively
different alternatives for moving a unit of mass from regular (non-bank) bin $i$ of
histogram $\pex$: a unit of mass can be moved either to a regular bin $j$ or
the bank bin of $\qex$.

In the first case, the total cost of transportation of a unit of mass is
exactly the ground distance $\dex_{ij} = D_{ij}$ between regular bins $i$
and $j$.

In the second case, the immediate cost of transportation to the
bank bin is $\dex_{i,bank} = \gamma$. However, because we have routed mass
from a regular bin
to the bank bin, there exists a regular bin $s$ in $\qex$ having ``mass deficit'' that has to be fulfilled from the bank bin of $\pex$.
Thus, if we move a unit of mass from a regular bin of $\pex$ to the bank bin of $\qex$, there
is an additional incurred cost $\gamma$ of moving an additional unit of mass
from the bank bin of $\pex$ to some regular bin of $\qex$. Hence, the total
cost of transportation of a unit of mass in the second case is
$$
    \gamma + \gamma = 2 \alpha \max\limits_{i,j}{D_{ij}} \geq (\text{since } \alpha \geq 0.5) \geq \max\limits_{i,j}{D_{ij}}.
$$

Thus, from the point of view of optimal mass transportation, it may never be
preferable to move a unit of mass from a regular bin to the bank bin
if there is an option to transport mass from a regular bin to a regular bin.
Consequently, an optimal solution to the \emdbanks's transportation problem can be
decomposed as follows.
\begin{align*}
\scriptsize
	&\memdbanks(P,Q,D)
        = \memd(\pex, \qex, \dex) \left(
    	    \sum\limits_{i=1}^{n}{P_i} +
    	    \sum\limits_{j=1}^{n}{Q_j}
        \right)\\
    & = \min\limits_{\{f_{ij}\}}{
            \sum\limits_{i,j=1}^{n+1}{f_{ij} \dex_{ij}}
        } = (\text{let } n + 1 = b)\\[0.1in]
    & = \min\limits_{\{f_{ij}\}}{}
            {\tiny
            \Big[
                \underbrace{
                    \sum\limits_{i,j=1}^{n}{ f_{ij} \dex_{ij} }
                }_{
                    \substack{\text{regular bins}\\\text{to regular bins}}
                }
                +
        		\underbrace{
        		    \sum\limits_{i=1}^{n}{ f_{ib} \dex_{ib} }
        		}_{
        		    \substack{\text{regular bins}\\\text{to bank bin}}
        		}
        		+
        		\underbrace{
    		        \sum\limits_{j=1}^{n}{ f_{bj} \dex_{bj} }
                }_{
                    \substack{\text{bank bin to}\\\text{regular bins}}
                }
                +
                \underbrace{
                    f_{bb} \dex_{bb}
                }_{
                    \substack{\text{bank bin}\\\text{to bank bin}}
                }
            \Big]
            }\\
	& = \min\limits_{\{f_{ij}\}}{
	    \Big[
    		\sum\limits_{i,j=1}^{n}{ f_{ij} D_{ij} } +
    		\gamma
    		\underbrace{
    		    \sum\limits_{j=1}^{n}{ f_{bj} }
    		}_{\Delta}
	    \Big]
    }
    = \min\limits_{\{f_{ij}\}}{
        \Big[
    		\sum\limits_{i,j=1}^{n}{ f_{ij} D_{ij} } +
    		\gamma \Delta
        \Big]
    }\\
    & = \min\limits_{\{f_{ij}\}}{ \Big[
		    \sum\limits_{i,j=1}^{n}{ f_{ij} D_{ij} }
	    \Big]} + \gamma \Delta
	  = \memdpelwer(P,Q,D).
\end{align*}
An additional useful observation is that a particular value of $k$ does
not matter for \emdbanks, since for every optimal solution of its underlying transportation problem, any amount of mass to the excess of $\Delta$ in the bank bin of the lighter
histogram will be transported at zero-cost $\dex_{bank,bank}$ to the bank bin of the heavier histogram.
This observation is formalized in the following corollary.
\end{proof}
\begin{corollary}
    For histograms $P=[P_1,\dots,P_n]$ and $Q=[Q_1,\dots,Q_n]$, and
    ground distance $D$, if $\sum{P_i} = \sum{Q_j}$ and $D$ is metric,
    then for all $k \geq 0 \in \mathbb{R}^+$, the following holds.
	\begin{align*}
		\memd\left([P, k], [Q, k], \left[\begin{array}{c|c}
            D & \begin{array}{c}|\\\omega\\|\end{array}\\ \hline
            \mbox{---}~\omega~\mbox{---} & 0
		\end{array}\right]\right) = \memd(P, Q, D),
	\end{align*}
	where $[X, k]$ is histogram $X$ extended with a single bank bin with capacity
	$k$ and a uniformly defined ground distance
	$\omega \geq \frac{1}{2} \max\limits_{i,j}{D_{i, j}}$ to/from the regular
	bins of $X$. In other words, if two histograms have equal total masses, we
	can increase their total masses by an arbitrary non-negative $k$ without
	affecting \emd between the histograms.
	\label{thm:bank-overflow}
\end{corollary}
\end{fullver}

Our version of Earth Mover's Distance, \emdnew, extends the idea of \emdbanks by
augmenting the histograms to even their masses. However, unlike \emdbanks, \emdnew extends
the histograms with multiple \emph{local bank bins} and distributes the total mass mismatch
over all of them. We, thereby, relate the mass mismatch penalty to the network
structure, while achieving equality of the total masses of the two histograms under comparison.

With respect to the number and location of the bank bins, one extreme option is to
attach one local bank bin to each initially present bin. Furthermore, in order to
model a non-constant transportation cost to/from a bank bin, multiple local bank
bins per initially present bin can be used, each with its individual ground distance.
This bank allocation strategy can incur a high computational cost, since attaching
even a single bank bin to each initially present bin doubles the size of
the histogram.

A compromise between the two extremes of having a single bank (as in \emdbanks) and
one bank per bin is to use one or more local banks per a \emph{group of bins}. Such bin
groups can be defined based on the
structural proximity of the corresponding users in the underlying network (see
 Fig.~\ref{fig:emd-new-local-banks-clusters}).
\begin{figure}[h!]
    \centering
	\includegraphics[width=3in]{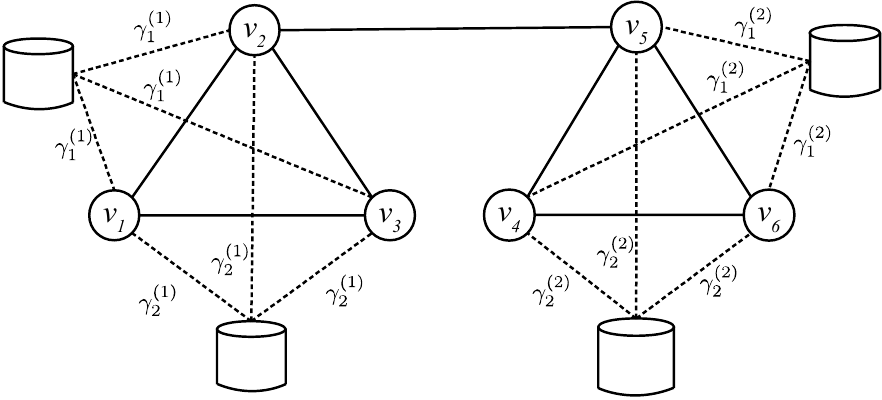}
	\caption{
		A histogram over a network extended with two banks per cluster of bins.
		$\gamma^{(i)}_j$ is the ground distance to/from the $j$'th bank bin of $i$'th bin cluster.
	}
	\label{fig:emd-new-local-banks-clusters}
	\vsb
\end{figure}

Ground distance $\gamma^{(i)}$ to/from an added bank bin should be of the same
order as the ground distances within the $i$'th cluster of bins the bank is
attached to. If $\gamma^{(i)}$ is much lower than the ground distances in
its cluster, then it can negatively affect metricity of \emdnew, the conditions
for which will be stated later. If $\gamma^{(i)}$ is much higher than the
ground distances in its cluster, it may result in an \emdbanks-like behavior,
with the global bank bin, even though spatially distributed across multiple local
banks, still playing no
role in the process of optimal mass transportation. The capacities of the added
bank bins should be determined based upon two ideas. Firstly, the capacity
of a bank bin should intuitively be proportional to the total mass of the bins
the bank is attached to, thereby, preserving the relative distribution of mass
over the network. Secondly, the capacities of all the bank bins should be such,
that the two histograms under comparison have equal total masses.
The following definition of capacity $P^{(i)}$ of a bank bin connected to the
$i$'th cluster $C_i$ of bins in the context of comparing histograms
$P=[P_1, \dots, P_n]$ and $Q=[Q_1, \dots, Q_n]$ incorporates both of the above requirements.
\vsa
\begin{align*}
	P^{(i)}_j = \begin{cases}
    	    \sum\limits_{v_k \in C_i}{P_k}
            ~\big/~
			\big(\sum\limits_{s=1}^{n}{Q_s} -
			\sum\limits_{s=1}^{n}{P_s}\big)
        & \text{if } \sum{Q_s} > \sum{P_s},\\
        \phantom{x}0, & \text{otherwise}.
    \end{cases}
\end{align*}

\vspace{-0.01in}
To better understand the advantage of \emdnew over the existing versions
of \emd, consider the example in Fig.~\ref{fig:emdnew-vs-other-emds}.
\begin{figure}
    \centering
	\includegraphics[width=3.3in]{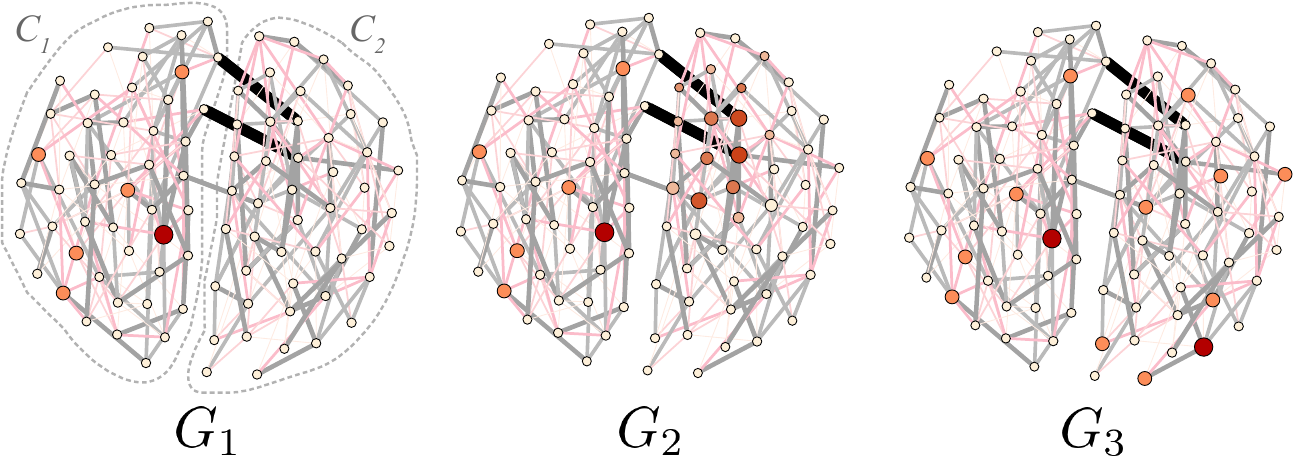}
	\vsc\vsa
	\caption{Three histograms over the same network.}
	\label{fig:emdnew-vs-other-emds}
	\vsa\vsa\vsc
\end{figure}
There are three histograms over a network with two pronounced clusters $C_1$ and
$C_2$ connected by three bridge edges. The distribution of mass over cluster $C_1$ is identical
in all three histograms, while cluster $C_2$ is empty in $G_1$ and has some differently
distributed mass in $G_2$ and $G_3$. In $G_2$ the extra mass has been ``propagated''
from cluster $C_1$ to cluster $C_2$ through the bridges, while in $G_3$ the same amount of extra
mass has been randomly distributed over cluster $C_2$. Thus, if we assume that $G_2$ and $G_3$
have ``evolved'' from $G_1$ through a process of mass propagation, then $G_2$ should
intuitively be closer to $G_1$ than $G_3$. However, only \emdnew captures this
intuition as $\memdnew(G_1, G_2) < \memdnew(G_1, G_3)$, while for \emdbanks and \emdpelwer,
$G_2$ and $G_3$ are equidistant from $G_1$, and for \emd, both $G_2$ and $G_3$ are identical to
$G_1$.

Next, we formally define \emdnew. Suppose we are given two histograms
$P = [P_1, \dots, P_n]$ and $Q = [Q_1, \dots, Q_n]$ defined over a network
$G = \tuple{V, E}$ with cross-bin ground distance $D_{n \times n}$. The ground
distance is application-specific and, in our case, is provided by \emdsent.
Bins of both histograms are clustered into groups $\{C_i\}$, $i = 1, \dots, N_c$
based on the proximity of the corresponding users in network $G$. Cluster $C_i$
contains $NC_i$ users/bins. Each bin cluster gets $N_b$ banks attached to all
its bins, so the total number of bins in an extended histogram is $N = n + N_c \times N_b$.
Ground distances from/to the bins of cluster $C_i$ to/from its banks are defined as
$\gamma^{(i)} = [ \gamma_1^{(i)}, \dots, \gamma_{N_b}^{(i)} ]^T$,
and, jointly for all clusters,
$\gamma = [ {(\gamma^{(1)})}^T, \dots, {(\gamma^{(N_c)})}^T ]^T$.
Since bank bins belonging to different clusters are not interconnected, in
order to define ground distances between bank bins \begin{changed}we
define distances between clusters $C_i$ in terms of $D_{ij}$ as
\end{changed} $d_{ij} = \min\limits_{v_p \in C_i, v_q \in C_j}{\{D_{pq}\}}$.
\begin{changed}Then, \emdnew is defined as follows.\end{changed}
{
\begin{align*}
	\pex &= \Big[
		\underbrace{ P_1, \dots, P_n}_{\text{original $P$}},
		\underbrace{ P^{(1)}_1, \dots, P^{(1)}_{N_b} }_{ \text{cluster $C_{1}$ banks} },
		\dots,
		\underbrace{ P^{(N_c)}_1, \dots, P^{(N_c)}_{N_b} }_{ \text{cluster $C_{N_c}$ banks} }
	\Big], \notag \\
	\qex &= \big[
		\underbrace{ Q_1, \dots, Q_n}_{\text{original $Q$}},
		\underbrace{ Q^{(1)}_1, \dots, Q^{(1)}_{N_b} }_{ \text{cluster $C_{1}$ banks} },
		\dots,
		\underbrace{ Q^{(N_c)}_1, \dots, Q^{(N_c)}_{N_b} }_{ \text{cluster $C_{N_c}$ banks} }
	\big], \notag \\[0.1cm] 
	S &= \left[\begin{array}{c}
			d_{1,*}    \otimes \onebb_{NC_1 \times 1} ~\big|~
		   		        \cdots                        ~\big|~
			d_{N_c, *} \otimes \onebb_{NC_{N_c} \times 1}
		\end{array}\right]^T, \notag \\[0.2cm] 
	\dex &= \left[\begin{array}{c|c}
		D_{n \times n}
		&
		\onebb_{n \times 1} \otimes \gamma^T + S^T \otimes \onebb_{1 \times N_b} \\
		\hline
		\substack{
			\onebb_{1 \times n} \otimes \gamma +
			S \otimes \onebb_{N_b \times 1}
		}
		&
		\substack{
			\gamma \otimes \onebb_{1 \times (N_b \cdot N_c)} +
			\gamma^T \otimes \onebb_{(N_b \cdot N_c) \times 1} -\\
			-\ 2 \cdot \mathop{diag}{(\gamma)} + d \otimes \onebb_{N_b \times N_b}
		}
	\end{array}\right], \notag 
\end{align*}
\vsb
\begin{align}
	\memdnew(P,Q) = \memd(\pex, \qex, \dex) \max{
	    \left\{
	        \sum{P_i}, \sum{Q_j}
	    \right\}
    },
\label{eq:emdnew-def}
\end{align}
}
where $\onebb_{a \times b}$ is an $a$-by-$b$ matrix of all ones; $d_{*,j}$ and $d_{i,*}$
are the $j$'th column and the $i$'th row of matrix $d$, respectively; $diag(v)$ is a diagonal matrix with the elements of vector $v$ on its main diagonal; and $\otimes$ is the Kronecker product.


Metricity of \emdnew, which can be exploited to improve practical
performance of distance-based search in applications~\cite{clarkson2006nearest}, is established in the following Theorem~\ref{thm:emdnew-metricity}.
\begin{confver}The proof is provided in the extended version of this
paper~\cite{snd-full}.\end{confver}

\ifconfver
    \vsb
\fi
\begin{theorem}
	Given an arbitrary finite set $\calh$ of histograms with bin clusters $\{C_i\}$
	and ground distance $D_{n \times n}$, if $D$ is metric and the ground distances
	$\gamma$ to/from the bank bins are such that
	$\forall i, j: \gamma^{(i)}_j \geq \frac{1}{2} \max_{v_p,v_q \in C_i}{D_{pq}}$,
	then \emdnew defined with $\{C_i\}$ and $\gamma$ is metric on $\calh \times \calh$.
\label{thm:emdnew-metricity}
\end{theorem}
\begin{fullver}
\begin{proof}
    Let us define constant $M = \max\limits_{X \in \calh}{\sum_{k}{X_k}}$.
	Since $\mathcal{H}$ is finite and all the distributions are assumed to
	have finite total masses, then $M < +\infty$. Next, we define an auxiliary
	distance measure $EMD'$ as follows.
	\begin{align*}
		&EMD'(P, Q, D) = \memd(P', Q', D'),\\[0.05in]
		P' = &[\pex, M - \sum{\pex_i}],\hspace{0.2in} Q' = [\qex, M - \sum{\qex_j}],\\[0.05in]
		D' = &\left[\begin{array}{c|c}%
			\dex & \begin{array}{c} |\\ \max\limits_{i,j}{\{\dex_{ij}\}} / 2\\ | \end{array}\\
			\hline
			\mbox{---}~\max\limits_{i,j}{\{\dex_{ij}\} / 2~\mbox{---}} & 0
		\end{array}\right],
	\end{align*}
	where $\pex$, $\qex$, and $\dex$ are the extended histograms, and the extended ground distance, respectively, as defined by \emdnew. From the definition of
	\emdnew~(\ref{eq:emdnew-def}), it follows
	that $\sum{\pex_i} = \sum{\qex_i}$ and, hence $M - \sum{\pex_i} = M - \sum{\qex_i} = k$.
	Thus, since $\sum{P'_i} = \sum{Q'_i} = M$, $D$ is metric, and $k \geq 0$, we can
	apply Corollary~\ref{thm:bank-overflow} to $P = \pex$, $Q = \qex$, $D = \dex$, and
	$\omega = \frac{1}{2}\max\limits_{i,j}{\dex_{ij}}$, to obtain
	\begin{align*}
		\memd'(P, &Q, D) =\ \memd(P', Q', D') =\\[0.05in]
		=\ &(\text{from Corollary~\ref{thm:bank-overflow}}) = \memd(\pex, \qex, \dex) =\\[0.05in]
			=\ &(\text{from definition of \emdnew}) =
    			\frac{
    		        \memdnew(P, Q, D)
    		    }{
    		        \max{ \left\{
	                    \sum{P_i}, \sum{Q_j}
	                \right\} }
    		    }.
	\end{align*}

	Thus, \emdnew is metric iff $\memd'$ is metric. The latter's metricity, according
	to Theorem~\ref{thm:emd-metric}, requires equality of total masses of all histograms
	and metricity of the ground distance. From the definition
	of $\memd'$, it is clear that \emph{all} histograms $P'$ and $Q'$ supplied to \emd
	by $\memd'$ always have the same total mass $M$. As to metricity of the ground distance,
	the identity of indiscernibles and symmetry straightforwardly follow from the corresponding
	properties of the original ground distance $D$ and our choice of the ground
	distances to/from the bank bins to be non-negative and symmetric. The triangle
	inequality holds for the original $D$, so we need to inspect only the new
	``triangles'' introduced into the ground distance after the addition of the
    bank bins, such as shown in Fig.~\ref{fig:emd-new-proof-triangle-general}.
\begin{figure}[h!]
    \centering
	\includegraphics[width=1.2in]{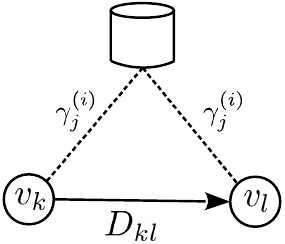}
	\caption{
        $j$'th bank bin of bin cluster $C_i$ attached to two representatives
        of $C_i$, namely, bins $v_k$ and $v_l$. Other bins of $C_i$ as well as
        other bank bins attached to it are not displayed.
        $D_{kl} \leq \max\limits_{v_p,v_q \in C}{D_{pq}}$.
	}
	\label{fig:emd-new-proof-triangle-general}
\end{figure}

	From the inequality for $\gamma^{(i)}_j$ in the theorem's statement,
    $\gamma^{(i)}_j + \gamma^{(i)}_j \geq 2 \times \frac{1}{2} \max\limits_{v_p,v_q \in C_i}{D_{pq}} \geq D_{kl}$, while $\gamma^{(i)}_j + D_{kl} \geq \gamma^{(i)}_j$
    trivially holds. Thus, the triangle inequality holds for each triangle introduced
    into the ground distance by the bank bin. When extending histogram $\pex$
    to $P'$, the same reasoning applies to the single added bank bin and ground
    distance $\dex$, and, as a result, the triangle inequality holds for $D'$
    as well. Thus, by Theorem~\ref{thm:emd-metric}, $EMD'$, and, hence, \emdnew
    is metric.
\end{proof}
\end{fullver}

\vsc
\section{Efficient Computation of SND}

\emdsent is defined~(\ref{eq:emdsent-def}) as a linear combination of several instances of \emdnew,
and, thus, computation of \emdsent involves:\\
\phantom{x} $\bullet$ Computing the ground distance $D$ based on the structure
of the underlying network $G = \tuple{V, E}$ ($|V| = n$, $|E| = m$) and the
opinions of the users in one of the network states $G_1$, $G_2$ under comparison.\\
\phantom{x} $\bullet$ Computing \emdnew, given that both the histograms
corresponding to the network states and the ground distance have been
computed.

Computation of the ground distance $D$ implies computing the shortest paths in
the network.
A direct computation of $D$ for all pairs of users using Johnson's
algorithm~\cite{johnson1977efficient} for sparse $G$ would incur time cost $\bigoh{n^2 \log{n}}$.
Computing \emdnew is algorithmically equivalent to computing \emd,
and, since the latter is formulated as a solution of a transportation
problem, it can be computed either using a general-purpose linear solver,
such as Karmakar's algorithm, or a solver that
exploits the special structure of the transportation problem, such as the
transportation simplex algorithm. The time complexity of both algorithms,
however, is supercubic in $n$. Thus, a precise computation of \emdsent
using existing techniques is prohibitively expensive at the scale of
real-world online social networks.
Furthermore, the existing approximations of \emd are either not applicable
to the comparison of histograms derived from the states of a social network,
since they drastically simplify the ground distance~\cite{tang2013earth,li2014linear},
or are effective only for \begin{changed}some graphs, such as trees,\end{changed} structurally not characteristic of
social networks~\cite{mcgregor2013sketching}.

We propose a method to compute \emdsent precisely in time linear in $n$
under the following two realistic \emph{assumptions}.

\textbf{Assumption 1:}
    The number $n_\Delta$ of users who change their opinion between
    two network states $G_1$ and $G_2$ under comparison is significantly
    smaller than the total number $n$ of users in the network. This assumption
    is reasonable, because in most applications the network states under
    comparison are not very far apart in time and, hence, $n_\Delta \ll n$.

\textbf{Assumption 2:}
    The opinion transportation costs, defined as the elements of 
    adjacency matrix $A^{ext}$ in~(\ref{eq:emdsent-extended-adjmat}), are positive
    integers bounded from above by constant $U \ll +\infty \in \mathbb{Z}^+$.
    This assumption is easy to satisfy by the appropriate choice of costs, and
    does not limit our analysis.

Since, according to the definition~(\ref{eq:emdsent-def}) of \emdsent, its computation
is equivalent to four computations of \emdnew, we will, first, focus on fast computation
of \emdnew on the inputs supplied by \emdsent.
Our method for efficient computation of \emdsent requires the following two lemmas.

\vspace{-0.18cm}
\begin{lemma}
    For any two histograms $P \in \mathbb{R}^n$ and $Q \in \mathbb{R}^m$ and
    ground distance $D \in \mathbb{R}^{n \times m}$, removal of empty bins
    from $P$ and $Q$ as well as the corresponding rows and columns from $D$
    does not affect the value of $\memdnew(P, Q, D)$.
\label{lemma:reduction1}
\end{lemma}
\vsa

The proof of Lemma~\ref{lemma:reduction1} is straightforward, since empty bins do not supply
or demand any mass in the underlying transportation problem, and, hence, do not
affect the cost of the optimal transportation plan. While Lemma~\ref{lemma:reduction1}
allows to remove redundant suppliers and consumers from the underlying transportation
problem, the following Lemma~\ref{lemma:reduction2} allows to transform the histograms,
without affecting the value of \emdnew, exposing the redundant suppliers and consumers
for removal.

\vsb
\begin{lemma}
Given two arbitrary histograms $P, Q \in \mathbb{R}^n$ and a ground distance
$D \in \mathbb{R}^{n \times n}$, if $D$ is semimetric\footnote{Under a
\emph{semimetric} we understand a metric with symmetry requirement dropped.
}, then for any $i \in [1; n]$, the following holds

\vsc
\begin{align*}
    &\memdnew(P, Q, D) = \memdnew( \\
    &\phantom{xxx}[P_1, \dots, P_{i - 1}, P_i - \min{\{P_i, Q_i\}}, P_{i+1}, \dots, P_n], \\
    &\phantom{xxx}[Q_1, \dots, Q_{i - 1}, Q_i - \min{\{P_i, Q_i\}}, Q_{i+1}, \dots, Q_n], D).
\end{align*}
\label{lemma:reduction2}
\end{lemma}
\vsc\vsbb
\begin{changed}
\begin{proof}
    First, we will show that there is always an optimal plan $f_{ij}$
    in the problem of optimal transportation of mass from $\pex$ to $\qex$ over
    $\dex$ such that $\forall i \in [1; n] : f_{ii} = \min{\{\pex_i, \qex_i\}} = M$, and,
    then, use such a plan to argue about the value of \emdnew.
    
    Consider an arbitrary optimal transportation plan $\fhat_{ij}$, and assume
    that $\exists i \in [1; n] : \delta = M - \hat{f}_{ii} > 0$. We will now use
    $\fhat$ to construct another optimal transportation plan $\newfhat_{ij}$ such
    that $\newfhat_{ii} = M$. Initially, we put $\newfhat = \fhat$ and, then, re-route
    mass flows in $\newfhat$ to eventually achieve the desired value of $\newfhat_{ii}$.

    Since, initially, $\newfhat_{ii} < M$, the remaining at least $\delta$ units of
    mass should be distributed by $\pex_i$ and consumed by $\qex_i$ to/from
    other consumers/suppliers. Among those, let us pick the ones that supply/consume
    the least amount of mass to $\qex_i$ and from $\pex_i$, respectively:
    $\ell = \argmin_{j \neq i}{\newfhat_{ji}}$, and $r = \argmin_{j \neq i}{\newfhat_{ij}}$. W.l.o.g., let us assume that $\newfhat_{\ell i} \leq \newfhat_{ir}$ and denote
    $\Delta = \min{\{\newfhat_{\ell i}, \delta\}}$. Now, we will re-route $\Delta$
    units of mass in $\newfhat$ as follows:
    \ifconfver
        \vspace{-0.18in}
    \else
        \vspace{0.05in}
    \fi
    \begin{figure}[h!]
        \centering
    	\includegraphics[width=1.15in]{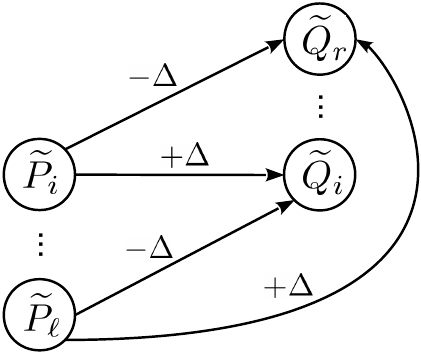}
    \vsb
    \end{figure}
    \vspace{-0.08in}
    \begin{align*}
        &\phantom{xxxx}\newfhat_{\ell i} \gets \newfhat_{\ell i} - \Delta,
        &\phantom{xxxx}\newfhat_{\ell r} \gets \newfhat_{\ell r} + \Delta, \phantom{xxxxx} \\
        &\phantom{xxxx}\newfhat_{ir} \gets \newfhat_{ir} - \Delta,
        &\phantom{xxxx}\newfhat_{ii} \gets \newfhat_{ii} + \Delta \phantom{xxxxxx}
    \end{align*}
    The updated transportation plan is legal, as the total amount of mass
    supplied or consumed by each bin has not changed. The total cost of $\newfhat$
    has been updated as follows
    \vspace{-0.03in}
    \begin{align*}
        new&cost(\newfhat) \gets cost(\newfhat) - \Delta \dex_{\ell i} - \Delta \dex_{i r}
            + \Delta \dex_{ii} + \Delta \dex_{\ell r} =\\
            =\ &\big(\text{since $D$ and, hence, $\dex$ is semimetric, $\dex_{ii} = 0$}\big) = \\
            =\ &cost(\newfhat) - \Delta(\dex_{\ell i} + \dex_{ir} - \dex_{\ell r}) \leq \\
            \leq\ &\big(\text{since $\dex$ is semimetric, $\dex_{\ell i} + \dex_{ir} \geq \dex_{\ell r}$} \big) \leq cost(\newfhat).
    \end{align*}
    Since the cost of the obtained legal plan $\newfhat$ cannot be less than the cost
    of an optimal plan $\fhat$, the performed update of $\newfhat$ has not changed its
    cost, and the updated $\newfhat$ is still an optimal plan. The described above
    re-routing procedure is repeatedly performed on $\newfhat$ until $\newfhat_{ii}$
    reaches $M = \min{\{\pex_i, \qex_i\}}$.
    
    Finally, to see why the statement of the lemma holds, we observe that the value
    of \emdnew is the cost of any optimal transportation plan, and the cost of $\newfhat$
    in particular. However, the cost of $\newfhat$ does not depend on $\newfhat_{ii}$,
    since, due to semimetricity of $\dex$, mass $\newfhat_{ii}$ gets transported at
    cost $\dex_{ii} = 0$. Thus, $M$ can be subtracted from $\pex_i$, $\qex_i$, and
    $\newfhat_{ii}$, without affecting the total cost of $\newfhat$. The solution of
    the latter modified transportation problem, however, is exactly
    \begin{align*}
        &\memdnew([P_1, \dots, P_{i - 1}, P_i - M, P_{i+1}, \dots, P_n],\\
        &\phantom{x}[Q_1, \dots, Q_{i - 1}, Q_i - M, Q_{i+1}, \dots, Q_n], D).
    \end{align*}
\end{proof}
\vsa
\end{changed}

Now, we will state our main result about the efficient computation of \emdsent
as Theorem~\ref{thm:emdsent-fast-comp}, whose constructive proof will describe the
computation method.

\vsb
\begin{theorem}
    Under Assumptions 1 and 2, \emdsent between network states $P = [P_1, \dots, P_n]$
    and $Q = [Q_1, \dots, Q_n]$ defined over network $G = \tuple{V, E}$, $(|V| = n, |E| = m)$
    can be computed precisely in time
    \vsa
    $$
        \bigoh{n_\Delta (m + n \sqrt{\log{U}} + n_\Delta^2 \log{(n_\Delta n U)})}.
    $$
\label{thm:emdsent-fast-comp}
\vsc\vsb
\end{theorem}
\begin{proof}
    We will focus on the efficient computation of the first summand $\memdnew(P^+, Q^+, D(P, +))$
    in the definition~(\ref{eq:emdsent-def}) of $\memdsent(P, Q, D)$,
    as computation of three other summands is procedurally equivalent
    and takes the same time.
    For the analysis of the computation of $\memdnew(P^+, Q^+, D(P, +))$, let us assume,
    without loss of generality, that
    $\sum_{i=1}^{n}{P_i^+} \geq \sum_{j=1}^{n}{Q^+_j}$. Let us also assume, for the
    ease of explanation, that the histograms are extended with one bank per bin.
    By definition~(\ref{eq:emdnew-def}), $\memdnew(P^+, Q^+, D(P, +))$ is the solution of a transportation
    problem with supplies $\pexp = [P^+_1, \dots, P^+_n, 0,\\ \dots, 0]$, demands
    $\qexp = [Q^+_1, \dots, Q^+_n, B_1, \dots, B_n]$ and ground distance $\dexpp$, where
    $B_i$ is the bank bin attached to bin $Q_i$, and the histograms and the ground distance
    extended according to the definition~(\ref{eq:emdnew-def}) of \emdnew.
    
    Now, we can apply Lemmas~\ref{lemma:reduction1} and \ref{lemma:reduction2} to reduce
    the size of the obtained transportation
    problem. From Assumption~2, $\dexpp$ is semimetric. Non-negativity and identity
    of indiscernibles straightforwardly follow from Assumption 2 and the definition
    of the length of a shortest path. Subadditivity follows from the shortest path
    problem's optimal substructure. Thus, we can apply Lemma~\ref{lemma:reduction2} to
    each pair $\pexp_i$, $\qexp_i$ of corresponding suppliers and consumers, and due
    to Assumption 1, a large number $(n - n_\Delta)$ of them have equal
    values. As a result, many suppliers and consumers become empty. Then, due to
    Lemma~\ref{lemma:reduction1}, all the obtained empty bins can be disregarded.
    If we put $M_i = \min{\{P^+_i, Q^+_i\}}$, then the reduced transportation problem
    is defined for suppliers
        $[P^+_{i_1} - M_{i_1}, \dots, P^+_{i_{n_\Delta}} - M_{i_{n_\Delta}}]$
    and consumers
        $[Q^+_{j_1} - M_{j_1}, \dots, Q^+_{j_{n_\Delta}} - M_{j_{n_\Delta}}, B_1, \dots, B_n]$,
    and ground distance $\dexpp$ that contains only the rows and columns corresponding to the remaining
    suppliers and consumers. The remaining suppliers and non-bank consumers are those that
    correspond to the users who have different opinion in $P^+$ and $Q^+$, and the number
    of such users, due to Assumption 1, is at most $n_\Delta$. The bank bins, however, do
    not get affected by Lemma~\ref{lemma:reduction2} (since only the banks of the lighter
    histogram can have non-zero mass) in $\qexp$ and hence do not get removed,
    yet they get removed from $\pexp$ due to Lemma~\ref{lemma:reduction1} as being empty.
    Thus, we have ended up with an unbalanced transportation problem, where the number
    $n_\Delta$ of suppliers is much less than the number $n + n_\Delta$ of consumers.
    
    Now, in order to compute $\memdnew(P^+, Q^+, D(P, +))$, we need to compute $\dexpp$
    and to actually solve the obtained transportation problem.
    
    Due to the structure of the reduced transportation problem, we need to compute only
    a small part of $\dexpp$. Since there are at most $n_\Delta$ suppliers, we need to
    solve at most $n_\Delta$ instances of single-source shortest path problem with
    at most $n_\Delta + n$ destinations. Since, due to Assumption 2, edge costs in the
    network are integer and bounded by $U$, each instance of a single-source shortest
    path problem can be solved using Dijkstra's algorithm based on a combination of a
    radix and a Fibonacci heaps~\cite{ahuja1990faster} in time
    \vspace{-0.05in}
    $$
        T_{sssp} = \bigoh{m + n\log{\sqrt{U}}}.
    $$
    \noindent (Notice, that if we assumed $\sum_{i=1}^{n}{P_i^+} \leq \sum_{j=1}^{n}{Q^+_j}$,
    and the reduced $\pexp$ contained $n_\Delta + n$ bins, we would \emph{not} need to
    run $n_\Delta + n$ instances of Dijkstra's algorithm. Instead, we would invert the
    edges in the network and compute shortest paths in reverse, still performing only
    $n_\Delta$ single-source shortest path computations.)
    
    Next we approach the solution of the reduced transportation problem with known
    ground distances. This problem can be viewed as a min-cost network flow problem
    in an unbalanced bipartite graph. Since, due to Assumption 2, edge costs are integers
    bounded by $U$, our min-cost flow problem can be solved using Goldberg-Tarjan's
    algorithm~\cite{goldberg1987solving} augmented with the two-edge push
    rule of~\cite{ahuja1994improved} in time
    
    \vsc\vsa
    \begin{align*}
        T_{transp} = \bigoh{n_\Delta m + n_\Delta^3 \log{(n_\Delta \max_{i,j}{\dexpp_{ij}})}}.
    \end{align*}
    \vsc
    
    \noindent Since no shortest path has more than $(n-1)$ edge, and the edge costs are bounded
    by $U$, the expression for time simplifies to
    
    \vsc
    \begin{align*}
        T_{transp} = \bigoh{n_\Delta m + n_\Delta^3 \log{(n_\Delta n U)}}.
    \end{align*}
    \vsc

    Thus, the total time for computing $\memd(P^+, Q^+, D(P, +))$ and, consequently,
    $\memdsent(P, Q, D)$ is
    \begin{align*}
        T =\ &\bigoh{n_\Delta T_{sssp} + T_{transp}} =\\ &\bigoh{n_\Delta (m + n\log{\sqrt{U}} + n_\Delta^2 \log{(n_\Delta n U)})}.
    \end{align*}
\end{proof}
\vsb

Notice that, if the social network is sparse, \begin{changed}that is  $m = \bigoh{n}$\end{changed},
and the number of changes $n_\Delta$
is bounded, then, according to Theorem~\ref{thm:emdsent-fast-comp}, \emdsent is
computable in time $\bigoh{n}$.

\vsbb
\section{Experimental Results}
In this section, we report experimental results, demonstrating the utility of
\emdsent in applications and comparing it with other distance measures. We
also study the scalability of our implementation of \emdsent.

\vsa
\subsection{Experimental Setup}

\textbf{Real-World Data:} Our Twitter dataset, based on data
from~\cite{macropol2013act} contains 10k users, each having an
average of 130 follower-followee edges. Among the tweets sent by these users
between May-2008 and August-2011, we select those relevant to political topics,
such as \emph{``Obama''}, \emph{``GOP''}, \emph{``Palin''},
\emph{``Romney''}. We break the entire observation period into quarters and, in each
quarter, assess every user's polarity with respect to our topics of interest.
Polarity of all users within one quarter comprise a network state.

\textbf{Synthetic Data:} We also perform experiments on synthetic 
scale-free networks of sizes $|V|$ from 10k to 200k and scale-free exponents from
$-2.9$ to $-2.1$. To generate the first network state, a
number of initial adopters are chosen uniformly at random, and approximately equal
numbers of them adopt \sentpos and \sentneg opinions. Each subsequent network state
$G_{i+1}$ is randomly generated from the preceding network state $G_i$ as follows.
A number of $G_i$'s neutral users get a chance to be activated. Each
of them adopts an opinion from her neighbors with probability
$\mathbb{P}_{nbr}$ and a random opinion with probability $\mathbb{P}_{ext}$.
If a user is to adopt an opinion from the neighbors, which opinion to adopt is decided
in a probabilistic voting fashion based on the numbers of active in-neighbors of each
kind.

\textbf{Distance Measures:} We compare \emdsent with the following distance
measures.\\
\phantom{x}{\emph{$\bullet\ \text{hamming}(P, Q)$}. The Hamming distance is a
        representative of all the distance measures performing coordinate-wise comparison.
    }\\
\phantom{x}{\emph{$\bullet\ \text{quad-form}(P, Q, L) = \sqrt{(P-Q)L(P-Q)^T}$.}
Quadratic-Form Distance~\cite{hafner1995efficient} based on the Laplacian $L$~\cite{merris1994laplacian}
        of the network. It takes the differences of opinions of the corresponding users and combines them based on the network's structure
    }\\
\phantom{x}{\emph{$\bullet\ \text{walk-dist}(P, Q) = \frac{1}{n}\|cnt(P) - cnt(Q)\|_1$}. Compares vectors $cnt(P) = [cnt(P_1), \dots, cnt(P_n)]$
        of users' ``contention'', where $cnt(P_i)$ is the amount by which the $i$'th user's
        opinion deviates from the opinion of this user's average active in-neighbor. Thus,
        \emph{walk-dist} summarizes how different the network's users
        are from their respective neighbors.
    }
    
\subsection{Detecting Anomalous Network States}

\noindent{\textbf{Synthetic Data:}} In a series of network states, we want to detect which ones are anomalous.
In particular, we are interested in the anomalies which are hard to detect
by observing the summary of the social network (e.g., the number of new
activations). Thus, in experiments with synthetic data, to simulate an anomaly,
we change the values of $\mathbb{P}_{nbr}$ and $\mathbb{P}_{ext}$ preserving
their sum, thereby, affecting only qualitatively the process of new users'
activation. In a series of network states, we compute the distances between
the adjacent states, normalize these distances by the number of active users,
and scale.
Then, spikes in the resulting series of distances are considered anomalies.

A qualitative analysis of anomaly detection on synthetic data is presented
in Fig.~\ref{fig:anomaly-synthetic}.
\vsa
\begin{figure}[h!]
\centering
\includegraphics[width=0.48\textwidth]{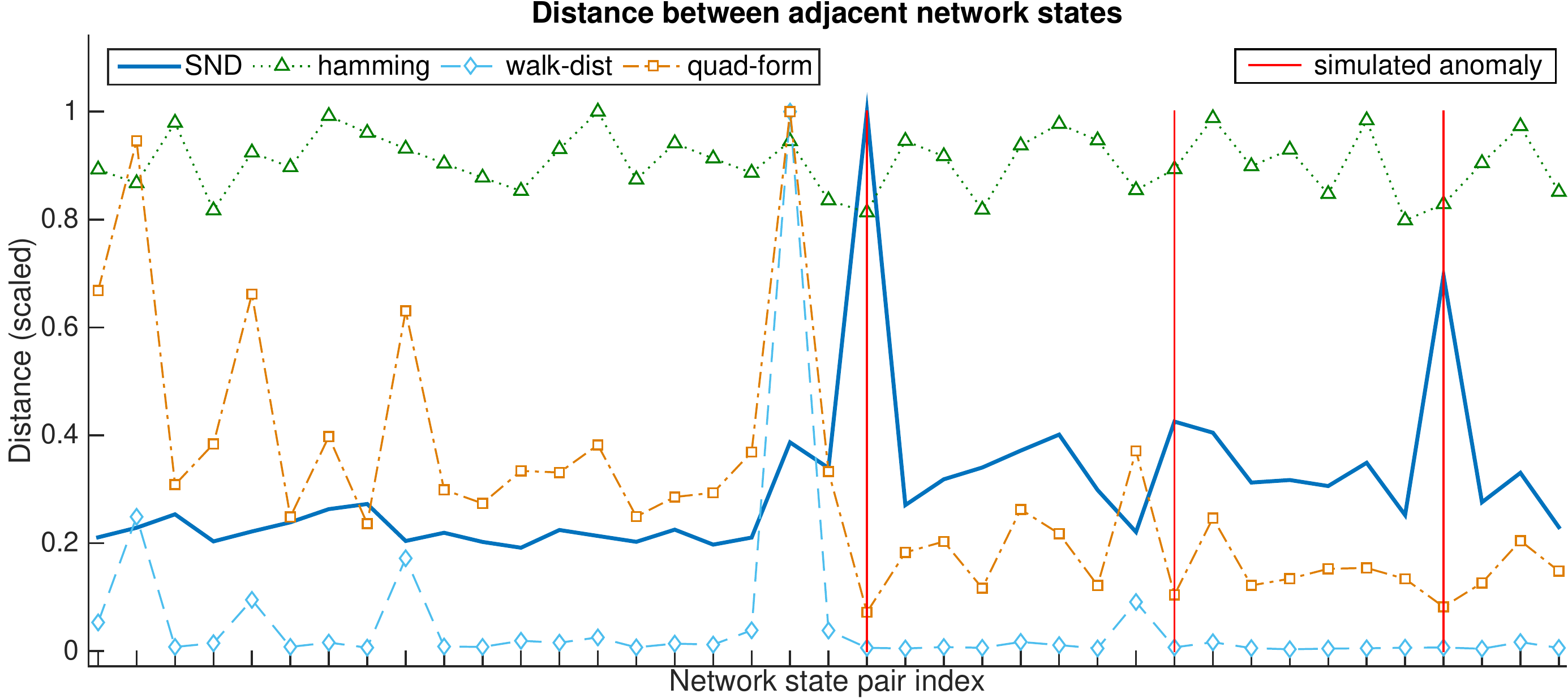}
\vsc\vsa
\caption{Anomaly detection on synthetic data. $|V| = 20\text{k}$, scale-free exponent
$\gamma = -2.3$. A series of $40$ network states is generated using
$\mathbb{P}_{nbr} = 0.12$ and $\mathbb{P}_{ext} = 0.01$ for normal and
$\mathbb{P}_{nbr} = 0.08$ and $\mathbb{P}_{ext} = 0.05$ for anomalous
network states' generation.
}
\label{fig:anomaly-synthetic}
\vsb
\end{figure}
For each simulated anomaly, \emdsent
produces a well noticeable spike, while other distance measures do not
recognize such anomalies.

In order to quantify the performance of the competing distance measures at
detecting the true simulated anomalies, we create a simple anomaly score
$S_t = (d_t - d_{t - 1}) + (d_t - d_{t + 1})$,
where $d_t$ is the value of a given distance measure at time $t$ normalized by the
number of users active at time $t$ and scaled. We rank the network state transitions
for each compared distance measure by $S_t$ in decreasing order and compute true and false positive rates for increasing ranks.
The corresponding ROC curves are displayed in Fig.~\ref{fig:anomaly-synthetic-roc}.
\vsb
\begin{figure}[h!]
\centering
\includegraphics[width=0.3\textwidth]{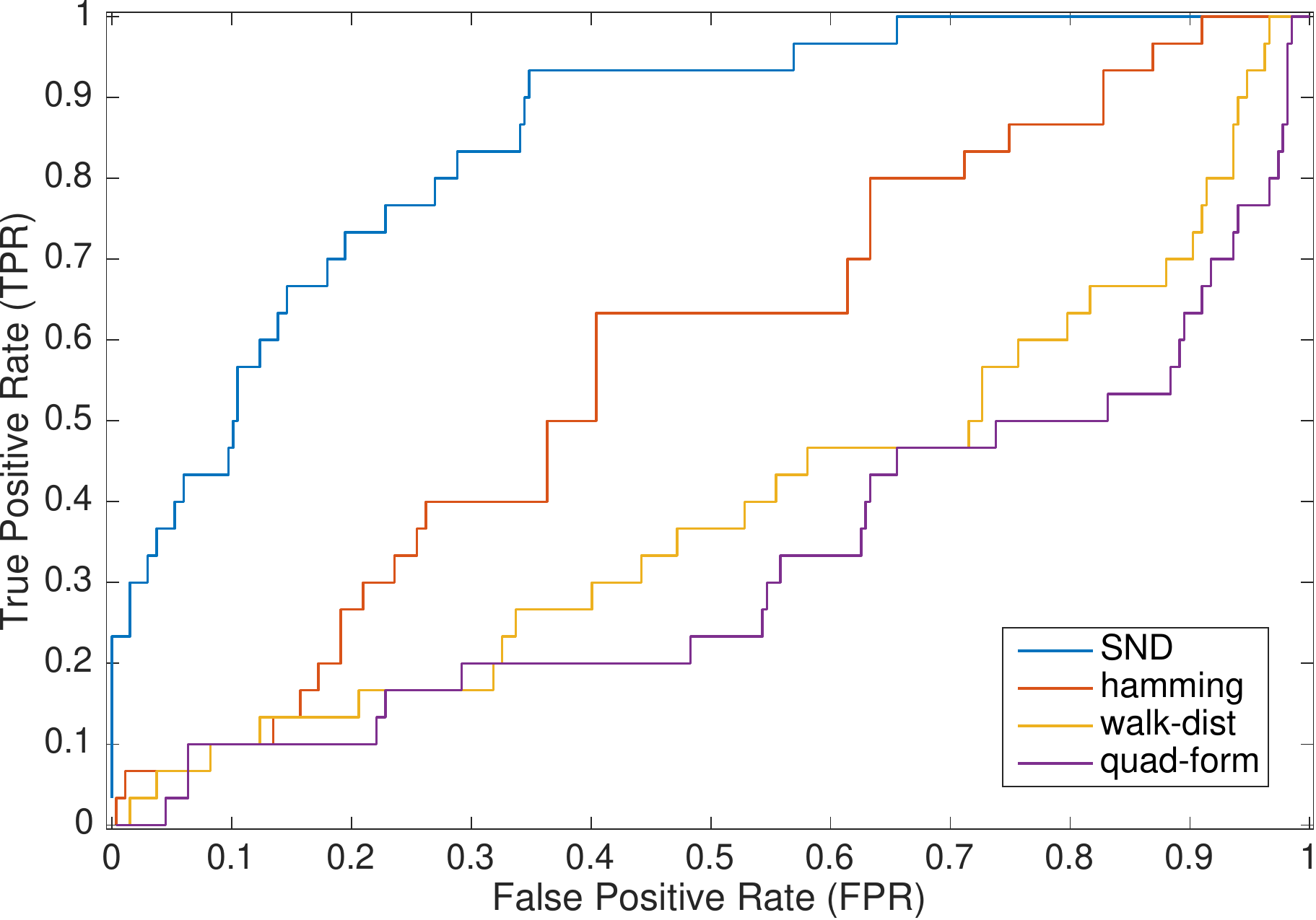}
\vsb
\caption{ROC curves comparing the quality of anomaly detection by different
distance measures in a series of $300$ network states over synthetic network
with $|V| = 30k$ and scale-free exponent $\gamma = -2.3$. The network states are
generated using $\mathbb{P}_{nbr} = 0.08$ and $\mathbb{P}_{ext} = 0.001$ for
normal and $\mathbb{P}_{ext} = 0.011$ and $\mathbb{P}_{nbr} = 0.07$ for
anomalous instances.}
\label{fig:anomaly-synthetic-roc}
\vsb
\end{figure}
\emdsent's accuracy dominates that of competing distance measures throughout the
spectrum of false positive rates. Particularly, for false positive rates up to $0.3$,
\emdsent achieves a true positive rate of $0.83$, while the next best distance measure
(\emph{hamming}) achieves only $0.4$.

\noindent{\textbf{Twitter Data:}} To obtain the ground truth for anomaly detection on our Twitter dataset, we collect ``search interest'' data from Google Trends\footnote{\scriptsize\url{http://www.google.com/trends/explore}}
and cross-check this data with American Presidents\footnote{\scriptsize\url{http://www.american-presidents-history.com}}
log of important political
events in the US. The anomaly detection results for topic ``Obama'' are shown
in Fig.~\ref{fig:anomaly-obama}.
\begin{figure}[h!]
\centering
\includegraphics[width=0.48\textwidth]{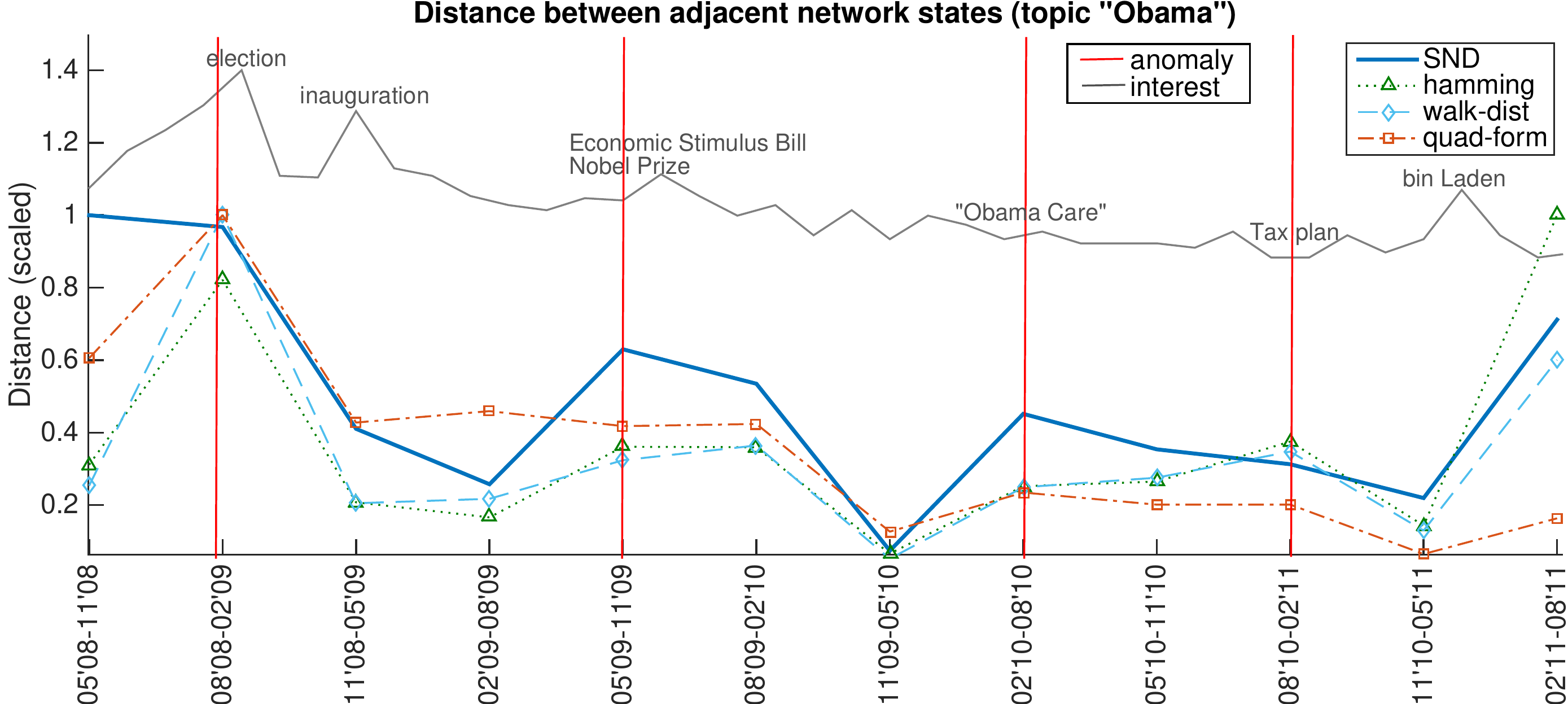}
\vsc
\caption{Anomaly detection on Twitter data (May'08-Aug'11) for topic ``Obama''.
The plots for distances between network states are accompanied by the plot showing
Google Trends' (scaled) interest in topic ``Obama''. Network states detected to
be anomalous by at least one distance measure are indicated with red vertical lines.}
\label{fig:anomaly-obama}
\vsbb
\end{figure}

We can distinguish two types of network states and, hence, events based on
\emdsent's behavior relatively to other distance measures. One type is the
events corresponding to network states where \emdsent agrees with other
distance measures. Two examples are (i) the first anomaly -- Barack Obama's election for the President of the US, and (ii) bin Laden's death being the
last spike on the Google Trends curve (even though, all distance measures noticeably increase
their value during the last quarter, we do not mark this quarter as anomalous,
since we do not have the distance values for the next quarter.) These events
are unlikely to have been perceived differently by the US users of Twitter
and, hence, probably have not provoked a polarized response.

\begin{changed}
The other type
of events are those where \emdsent noticeably disagrees with other distance
measures. For example, during quarters 05'09-11'09, the Economic Stimulus
Bill had a highly polarized response in the House of Representatives\footnote{\scriptsize
http://www.nytimes.com/2009/01/29/us/politics/29obama.html}, with no Republican
voting in its favor.
Another such anomaly takes place during quarters 02'10-08'10, when the
Affordable Care Act (``Obama Care'') was introduced, and which was and still remains
a very controversial topic. The latter can be seen from the House vote
distribution\footnote{\scriptsize Democrats -- 219 yeas, republicans -- 212 nays (\url{http://www.healthreformvotes.org/congress/roll-call-votes/h165-111.2010})}
and from the fact that, according to
socialmention.com\footnote{\scriptsize \url{http://socialmention.com/search?q=\%22obama+care\%22}},
even in October 2015, ``Obama Care'' is still perceived equally positively and
negatively in microblogs.
\end{changed}

\vsa
\subsection{Predicting User Opinions}

Given a series of states of a social network we want to predict the unknown
opinions of the users in the current network state $G_0$ based on the observed recent
$G_{-t}$ $(t \in \mathbb{N})$ and the (incomplete) current network states. For example,
if certain Twitter users have not tweeted (enough) in the current quarter, we may want
to predict the opinions of these users in the current quarter based on the observed
opinions of all users of the network. We assume that during the periods corresponding
to the observed recent network states $G_{-t}$, the social network evolved ``smoothly'',
that is, the recent states can predict the current state. Under this assumption, we use
a distance measure to compute the distances $dist(G_{-t}, G_{-t+1})$ between the adjacent
past network states, then, extrapolate the obtained series to estimate the distance $d^*$
from the most recent $G_{-1}$ to the yet unknown \emph{complete} current network state. Then, we assign different
opinions to the target users in the current network state, trying to make the distance
$dist(G_{-1}, G_0^*)$ from the most recent to the modified current network state as close
to estimate $d^*$ as possible. The search for the best assignment of opinions to the
target users is randomized --- the number of the uniformly randomly generated opinion assignments
for all target users is considerably lower than the total number of possible assignments
(we use $100$ random opinion assignments in each experiment). In each experiment, we
uniformly randomly select $20$ active users \begin{changed}-- with approximately equal number of positive and
negative users --\end{changed} in the current network state, predict their
opinions and measure the prediction accuracy. This procedure is repeated $10$ times, and
mean accuracies and standard deviations are reported.

\begin{changed}
The predictions are made using the above distance-based method with \emdsent
as well as other distance measures. To put the prediction performance of these methods
in context, we add into comparison two non-distance-based methods -- one basic, and one
state-of-the-art -- that make predictions based on
the known quantified opinions of the users and the network's structure. One such method,
\emph{nhood-voting}, derives the opinion of each target user based on the opinions of this
user's active in-neighbors in a probabilistic voting fashion, or selects it uniformly
randomly in the absence of active in-neighbors. Another method,
\emph{community-lp}~\cite[IV.B]{conover2011predicting} detects communities in the
network via label propagation and, then, predicts user opinions based on these users'
membership in the discovered communities.
\end{changed}

We experiment on both synthetic and real-world data. For synthetic data, we generate a
scale-free network with $n=10\text{k}$ users and scale-free exponent $\gamma = -2.5$.
A series of network states is generated using the same algorithm as in the case with
anomaly detection, with probabilities of opinion adoption from the neighborhood and
from the ``external source'' ranging between $0.001$ and $0.2$. The number of initial
adopters in the first network state is $800$. We use $3$ most recent network states
to estimate the distance from the most recent \begin{changed}to the incomplete current network state\end{changed}.

\begin{changed}
Results for opinion prediction are summarized in Table~\ref{tab:opinion-prediction-accuracy}.
There are three important observations:

$\bullet$ Firstly, among the distance-based methods, the one that uses \emdsent always
performs best, with an average prediction accuracy of 74-75\% and a consistently low
standard deviation. This suggests that \emdsent captures more opinion dynamics-specific
information than the other distance measures, and should be preferred, particularly, when
such simple statistics as the rate of new user activation are uninformative.
    
$\bullet$ Secondly, \emdsent-based prediction method works considerably better than method
\emph{nhood-voting} that bases the opinion prediction for each user on the opinions of the user's in-neighbors. This emphasizes the importance of analyzing opinion dynamics at the
level of the entire network, rather than for each egonet in isolation.
    
$\bullet$ Lastly, \emdsent-based method outperforms the state-of-the-art method
\emph{community-lp}, based on community detection via label propagation. In our
experiments, \emph{community-lp}'s prediction accuracy is 57-65\%, while this
method's authors report the accuracy of 95\% for
their data~\cite{conover2011predicting}.
The likely cause of such a discrepancy is a \emph{very} high level of homophily in their data
(users almost exclusively follow those having the same opinion), while in our less
homophilous data, \emph{community-lp} performs worse by capturing only users' reachability
by the opinions of each kind, and \emdsent performs better by looking for the \emph{most
probable} opinion propagation scenario.
\end{changed}
\begin{table}[ht]
\centering
\footnotesize
\begin{tabular}{| l | r | r || r | r |}
\hline
\multicolumn{5}{|c|}{\rule{0pt}{2.3ex} \textbf{User Opinion Prediction Accuracy, \%}} \\ \hline
\rule{0pt}{2.3ex}\multirow{2}{*}{\phantom{xxl}Method} & \multicolumn{2}{c||}{Synthetic Data} & \multicolumn{2}{c|}{Real-World Data} \\ \cline{2-5}
    & \multicolumn{1}{c|}{$\mu$} & \multicolumn{1}{c||}{$\sigma$} & \multicolumn{1}{c|}{$\mu$} & \multicolumn{1}{c|}{$\sigma$} \\ \hline \hline
\rule{0pt}{2.3ex}\emdsent   & \textbf{74.33} & \textbf{\phantom{0}2.65} & \textbf{75.63} & \textbf{5.60} \\ \hline
\rule{0pt}{2.3ex}hamming    & 68.44 & 12.34 & 68.13 & 5.80 \\ \hline
\rule{0pt}{2.3ex}quad-form  & 66.67 & 13.58 & 67.50 & 9.63 \\ \hline
\rule{0pt}{2.3ex}walk-dist & 56.22 & 15.35 & 31.88 & 9.98 \\ \hline \hline
\rule{0pt}{2.3ex}nhood-voting & 62.11 & \phantom{0}8.58 & 61.25 & 5.82 \\ \hline
\rule{0pt}{2.3ex}\begin{changed}community-lp\end{changed}& 65.25 & \phantom{0}9.43 & 56.87 & 8.43 \\ \hline
\end{tabular}
\caption{Means $\mu$ and standard deviations $\sigma$ of user opinion prediction accuracies.}
\label{tab:opinion-prediction-accuracy}
\vsc
\end{table}

\vsa
\subsection{Sensitivity to Opinion Dynamics Models}
In this section, we show the effectiveness of \emdsent in detecting qualitative
changes in the network's evolution w.r.t. an advanced opinion dynamics model,
that cannot be spotted by the distance measures performing coordinate-wise comparison.
We generate a number of pairs $\tuple{G_1, G_2}$ of subsequent
network states over a synthetic scale-free network. Some of these pairs correspond
to \emph{normal transitions}, while others correspond to \emph{anomalous transitions}
in the network's evolution. For the normal transitions, $G_2$ is generated from
$G_1$ using the Independent Cascade Model with Competition (ICC)~\cite{carnes2007maximizing}.
For anomalous transitions, most new activations in $G_2$ happen randomly, not relying
on the network's structure. We study the distances assigned to normal and anomalous
transitions by \emdsent and $\ell_1$, and plot them as functions of the number
$n_\Delta$ of users having different opinion in $G_1$ and $G_2$ of each transition.
According to the results in Fig.~\ref{fig:snd-vs-ell1}, \emdsent clearly separates
anomalous transitions from normal ones. $\ell_1$, however, cannot discern anomalous
transitions,
as its value is mostly determined by $n_\Delta$, which is representative of the
distance measures performing coordinate-wise comparison.

\vsbb
\begin{figure}[h]
  \begin{center}\hspace{-0.5cm}
    \mbox{
      \subfigure {\includegraphics[width=0.25\textwidth]{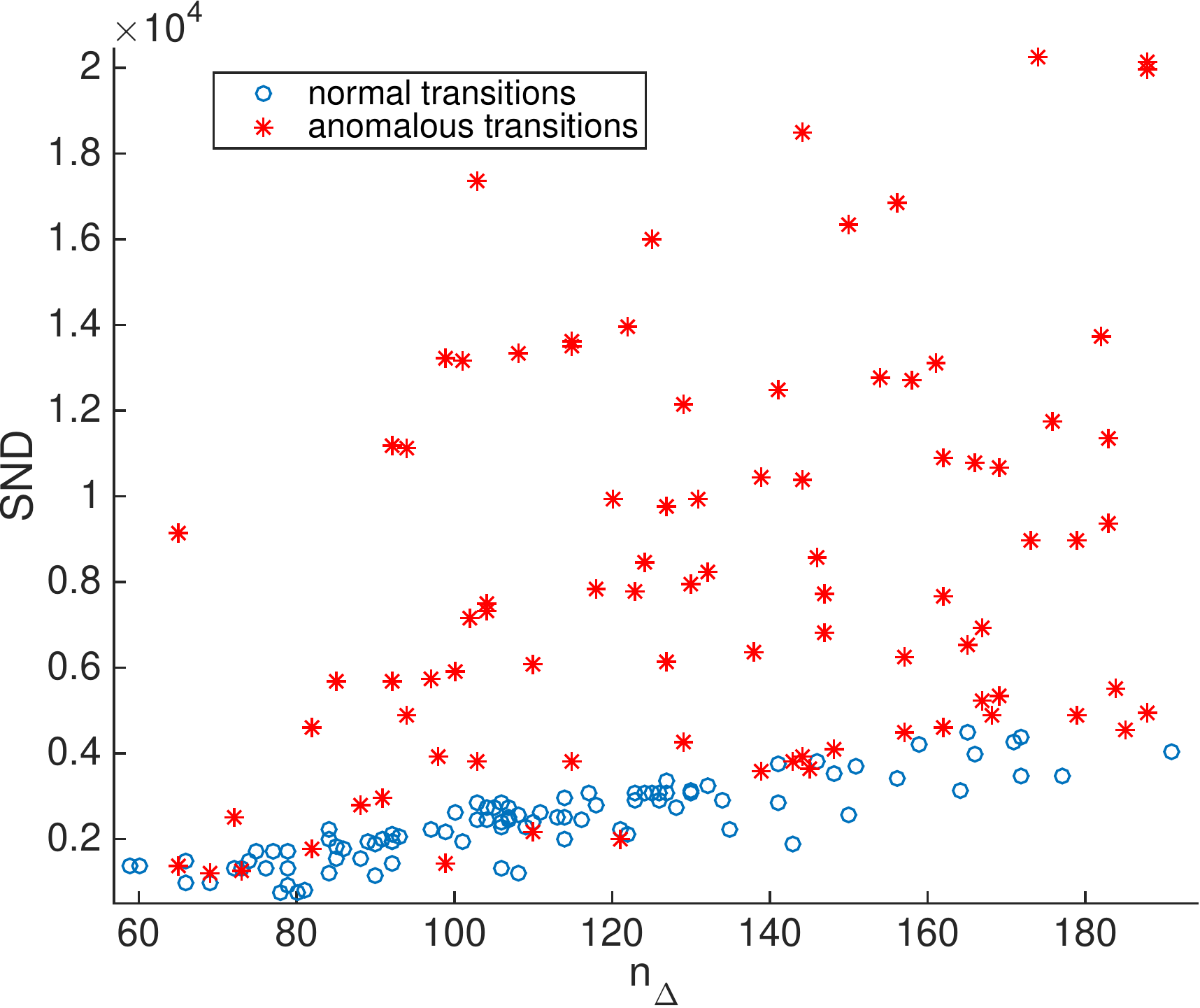}}\hspace{-0.3cm}
      \subfigure {\includegraphics[width=0.25\textwidth]{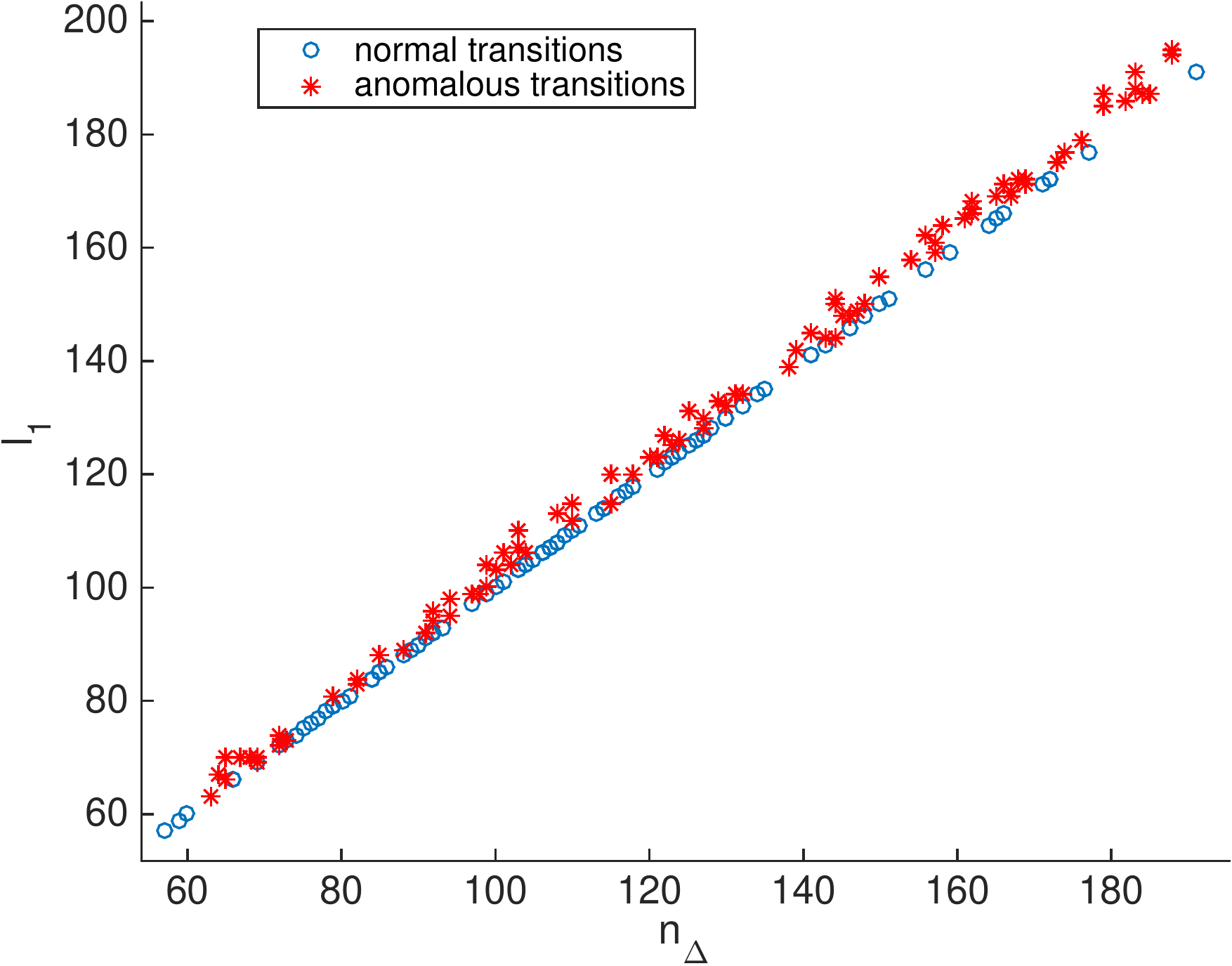}}
    }
    \vsc
    \caption{\emdsent and $\ell_1$ distances between network states of normal (ICC)
and anomalous (random) transitions.}
    \label{fig:snd-vs-ell1}
  \end{center} 
\end{figure}

\ifconfver
    \vsc\vsc
\else
    \phantom{x}
\fi
\subsection{Scalability of SND}
We implemented\footnote{\url{https://cs.ucsb.edu/~victor/pub/ucsb/dbl/snd/}}
\emdsent in MATLAB, with parts written in C++. We use a min-cost network flow
solver CS2~\cite{goldberg1997efficient} that implements Goldberg-Tarjan's
algorithm~\cite{goldberg1987solving}, but, unlike it is prescribed by
Theorem~\ref{thm:emdsent-fast-comp}, does not use the two-edge push rule
of~\cite{ahuja1994improved}. Additionally, for computing shortest paths,
our implementation of Dijkstra's algorithm uses a priority queue based
on a binary heap, rather than a combination of a Fibonacci and a radix
heaps~\cite{ahuja1990faster}. As a result, our implementation of \emdsent
scales slightly worse than linearly as guaranteed by Theorem~\ref{thm:emdsent-fast-comp},
but still very well to be applicable to real-world social networks.
Fig.~\ref{fig:emd-scaling-fixed-ndiff} shows how our implementation
of \emdsent based on the proposed efficient method scales in the number
$n$ of users in the network in comparison with a direct computation
of \emdsent based on a CPLEX linear solver. Our implementation's
scalability in the number $n_\Delta$ of users who have changed their
opinion is shown in Fig.~\ref{fig:emdsent-scaling-fixed-n}.
\vsc
\begin{figure}[h]
\centering
\begin{minipage}{0.22\textwidth}
    \includegraphics[width=\textwidth]{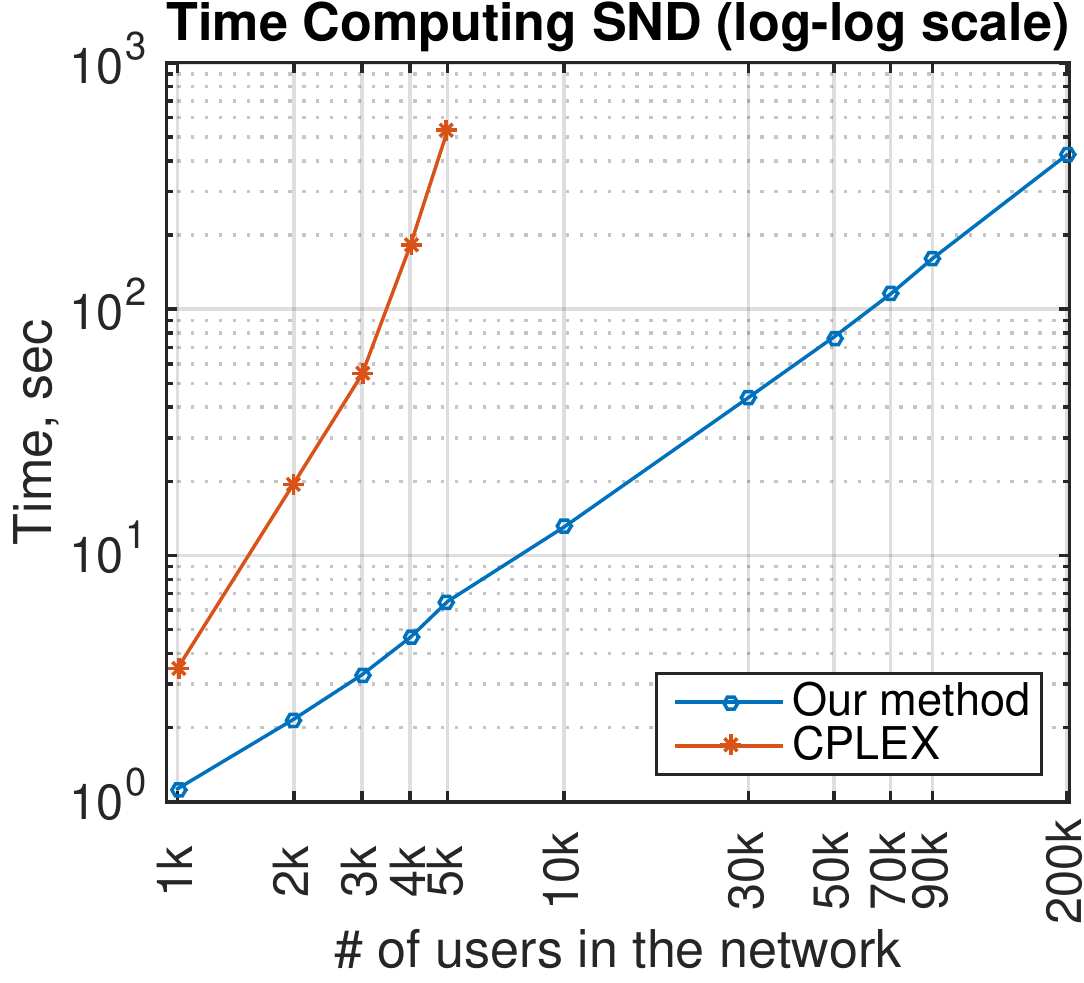}
    \vsc
    \caption{Time for computing \emdsent when the number of users having different opinion
    is fixed at $n_\Delta = 1000$ and the total number of users $n$ in the network grows up 200k.}
    \label{fig:emd-scaling-fixed-ndiff}
    \vsc
\end{minipage}
\quad
\begin{minipage}{0.22\textwidth}
    \includegraphics[width=\textwidth]{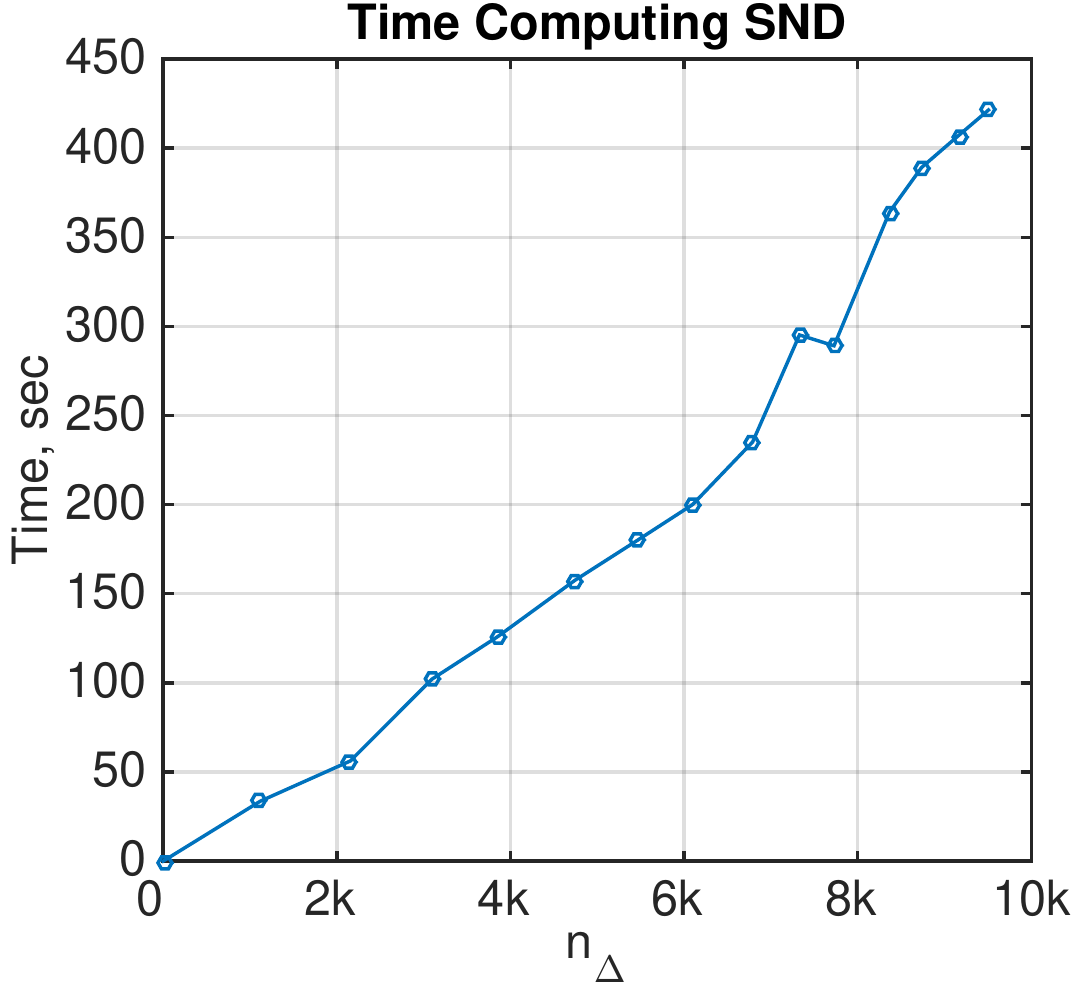}
    \vsb\vsa
    \caption{Time for computing \emdsent using our method when the network size is
    fixed at $n=20\text{k}$, and the number of users $n_\Delta$ having different
    opinions grows up to 10k.}
    \label{fig:emdsent-scaling-fixed-n}
\end{minipage}
\vsbb
\end{figure}

\iffullver
    \vsb
\else
    \vspace{0.05in}
\fi
\section{Related Work}
There is a large number of existing distance measures used in vector spaces,
including $\ell_p$, Hamming, Canberra, Cosine, Kullback-Leibler, and
Quadratic Form~\cite{hafner1995efficient} distances. However, none of them is
adequate for the comparison
of network states, since these distance measures either compare vectors
coordinate-wise, thereby, not capturing the interaction between users
in the network, or in the case of Quadratic Form distance, capture the user
interaction in a very limited and uninterpretable way.

Existing graph-oriented distance measures are also unsuitable for comparing
network states with polar opinions. The first class of such distance measures is graph isomorphism-based distance
measures, such as \emph{largest common subgraph}~\cite{bunke1998graph}. These
distance measures are node state-oblivious, and, hence are not applicable to
the comparison of network states.
Another class of graph distance measures is \emph{Graph Edit Distance
(GED)-based} measures~\cite{gao2010survey} that define the distance between
two networks as the cost of the optimal sequence of edit operations, such as node or
edge insertion, deletion, or substitution, transforming one network into another. GED can be
node state-aware, but its value is not interpretable from the
opinion dynamics point of view, and even its approximate computation takes
time cubic in $|V|$ \begin{changed}(a single computation of GED on a 10k-node network on
our hardware takes about a month)\end{changed}.

A third class of distance measures includes \emph{iterative distance measures}~\cite{blondel2004measure,
leicht2006vertex,melnik2002similarity}, which
express similarity of the nodes of two networks recursively, use a fix-point iteration
to compute node similarities, and, then, aggregate node similarities to obtain the 
similarity of two networks. Iterative distance measures share the problem of GED -- they do
not capture the way opinions spread in the network.

The last class includes \emph{feature-based} distance
measures \cite{berlingerio2012netsimile, wilson2005pattern, zhu2005study},
which compare either the distributions of local node properties (e.g., degree,
clustering coefficient) or the spectra of two networks.
Despite their efficient computability, such distance measures do not fit the
comparison of network states with polar opinions. The spectral distance
measures are inadequate because they do not deal with node states directly\footnote{Even
if node states
are artificially encoded into a network's structure, there is still
a possibility for two structurally different networks to have identical
spectra and, hence, a zero spectral distance.},
while other feature-based distance measures only deal with summaries based
on opinion of each kind and, thus, cannot capture the competition of
polar opinions.

\vsb
\begin{changed}
\section{Limitations}
Despite the demonstrated effectiveness and efficiency of \emdsent, there are scenarios
in which its use is either prohibitively or unnecessarily expensive.

$\bullet$ One reason to choose a simpler distance measure,
such as $\ell_p$, over \emdsent is the latter's computational cost. While it is asymptotically linear
in the number of nodes, it can, potentially, be too high in practice for networks having
100M+ nodes, where a single computation of \emdsent can take several days. If the
use of a simpler distance measure is undesirable, one can partition the network into
clusters of tractable size and perform the \emdsent-based analysis on each cluster.

$\bullet$ Another
scenario when using \emdsent may be excessive is when the changes in the rate of
new user activation reveal enough information for the target application (for example,
the activation rate alone is clearly enough to detect the US presidential election day),
and, in such a case, the distance measures as simple as Hamming distance may suffice.

\vsb
\section{Future Research}
Among the directions for future research are the following.

\vspace{0.01in}
$\bullet$ Since \emdsent is, effectively, the first distance measure
designed specifically for the comparison of states of a social network containing
competing opinions, one potential future research direction is using \emdsent in
other applications operating in a metric space setting, such as network state
classification, clustering, and search.

$\bullet$ Additionally, it may be lucrative to combine \emdsent with non-distance-based
methods. Thus, in the method of~\cite{conover2011predicting} that predicts opinions
based on the content of the users' tweets, the objective function can be augmented with an
\emdsent-based term, thereby, performing opinion fitting at both the micro-level of each user
and the macro-level of the entire network.

$\bullet$ Finally, it may be fruitful to design a distance measure that would capture
changes in both the opinions of the users and the structure of the social
network simultaneously. Such a distance measure would be more computationally complex than \emdsent
due to the network alignment requirement, yet, useful for the comparison of network
states defined over very different networks.

\end{changed}

\ifconfver
    \vsb
\fi
\section{Conclusion}

In this paper, we proposed Social Network Distance (\emdsent)
-- the first distance measure for comparing the states of a social network
containing competing opinions. Our distance measure quantifies how likely
it is that one state of a
social network has evolved into another state under a given model of
polar opinion propagation.
Despite the high computational complexity of the transportation problem
underlying \emdsent, we propose a linear-time algorithm for its precise
computation, making \emdsent applicable to real-world online social networks.
We demonstrate the usefulness of \emdsent in detecting
anomalous network states
and predicting user opinions in both synthetic and real-world data,
where it consistently outperforms other distance measures.
Our anomaly detection method achieves a true positive rate (TPR) of $0.83$,
while the next best method's TPR is only $0.4$. The accuracy of \emdsent-based method
for user opinion prediction averages at $75.63\%$, which is $7.5\%$ higher than that of
the next best method. We also show that, unlike the distance measures performing
coordinate-wise comparison, \emdsent can detect qualitative changes in the network's
evolution pattern.

Our results emphasize
    the importance of taking into account user locations
    in the analysis of social networks,
and
    that the analysis of opinion dynamics at the level of an entire
    social network can provide more information about the network's
    evolution than the methods operating at the level of egonets.

\bibliographystyle{abbrv}
{
    \bibliography{paper}

\begin{thebibliography}{10}

\bibitem{ahuja1990faster}
R.~K. Ahuja, K.~Mehlhorn, J.~Orlin, and R.~E. Tarjan.
\newblock Faster algorithms for the shortest path problem.
\newblock {\em Journal of the ACM (JACM)}, 37(2):213--223, 1990.

\bibitem{ahuja1994improved}
R.~K. Ahuja, J.~B. Orlin, C.~Stein, and R.~E. Tarjan.
\newblock Improved algorithms for bipartite network flow.
\newblock {\em SIAM Journal on Computing}, 23(5):906--933, 1994.

\bibitem{berlingerio2012netsimile}
M.~Berlingerio, D.~Koutra, T.~Eliassi-Rad, and C.~Faloutsos.
\newblock {NetSimile: a scalable approach to size-independent network
  similarity}.
\newblock {\em arXiv preprint arXiv:1209.2684}, 2012.

\bibitem{blondel2004measure}
V.~D. Blondel, A.~Gajardo, M.~Heymans, P.~Senellart, and P.~Van~Dooren.
\newblock A measure of similarity between graph vertices: Applications to
  synonym extraction and web searching.
\newblock {\em SIAM review}, 46(4):647--666, 2004.

\bibitem{borodin2010threshold}
A.~Borodin, Y.~Filmus, and J.~Oren.
\newblock Threshold models for competitive influence in social networks.
\newblock In {\em Internet and Network Economics}, pages 539--550. Springer,
  2010.

\bibitem{bunke1998graph}
H.~Bunke and K.~Shearer.
\newblock A graph distance metric based on the maximal common subgraph.
\newblock {\em Pattern recognition letters}, 19(3):255--259, 1998.

\bibitem{carnes2007maximizing}
T.~Carnes, C.~Nagarajan, S.~M. Wild, and A.~Van~Zuylen.
\newblock Maximizing influence in a competitive social network: a follower's
  perspective.
\newblock {\em EC}, pages 351--360, 2007.

\bibitem{clarkson2006nearest}
K.~L. Clarkson.
\newblock Nearest-neighbor searching and metric space dimensions.
\newblock {\em Nearest-neighbor methods for learning and vision: theory and
  practice}, pages 15--59, 2006.

\bibitem{conover2011predicting}
M.~D. Conover, B.~Gon{\c{c}}alves, J.~Ratkiewicz, A.~Flammini, and F.~Menczer.
\newblock {Predicting the political alignment of Twitter users}.
\newblock In {\em SocialCom}. IEEE, 2011.

\bibitem{gao2010survey}
X.~Gao, B.~Xiao, D.~Tao, and X.~Li.
\newblock {A survey of Graph Edit Distance}.
\newblock {\em Pattern Analysis and applications}, 13(1):113--129, 2010.

\bibitem{goldberg1987solving}
A.~Goldberg and R.~Tarjan.
\newblock Solving minimum-cost flow problems by successive approximation.
\newblock In {\em Proceedings of the nineteenth annual ACM symposium on Theory
  of computing}, pages 7--18. ACM, 1987.

\bibitem{goldberg1997efficient}
A.~V. Goldberg.
\newblock An efficient implementation of a scaling minimum-cost flow algorithm.
\newblock {\em Journal of algorithms}, 22(1):1--29, 1997.

\bibitem{goyal2010learning}
A.~Goyal, F.~Bonchi, and L.~V. Lakshmanan.
\newblock Learning influence probabilities in social networks.
\newblock {\em WSDM}, pages 241--250, 2010.

\bibitem{hafner1995efficient}
J.~Hafner, H.~S. Sawhney, W.~Equitz, M.~Flickner, and W.~Niblack.
\newblock Efficient color histogram indexing for quadratic form distance
  functions.
\newblock {\em Pattern Analysis and Machine Intelligence, IEEE Transactions
  on}, 17(7):729--736, 1995.

\bibitem{johnson1977efficient}
D.~B. Johnson.
\newblock Efficient algorithms for shortest paths in sparse networks.
\newblock {\em Journal of the ACM (JACM)}, 24(1):1--13, 1977.

\bibitem{leicht2006vertex}
E.~Leicht, P.~Holme, and M.~Newman.
\newblock Vertex similarity in networks.
\newblock {\em Physical Review E}, 73(2):026120, 2006.

\bibitem{li2014linear}
L.~Li, M.~Ma, P.~Lei, X.~Wang, and X.~Chen.
\newblock {A linear approximate algorithm for Earth Mover's Distance with
  thresholded ground distance}.
\newblock {\em Mathematical Problems in Engineering}, 2014.

\bibitem{emd-banks}
V.~Ljosa, A.~Bhattacharya, and A.~K. Singh.
\newblock Indexing spatially sensitive distance measures using multi-resolution
  lower bounds.
\newblock {\em EDBT}, pages 865--883, 2006.

\bibitem{macropol2013act}
K.~Macropol, P.~Bogdanov, A.~K. Singh, L.~Petzold, and X.~Yan.
\newblock {I act, therefore I judge: Network sentiment dynamics based on user
  activity change}.
\newblock {\em ACM ASONAM}, pages 396--402, 2013.

\bibitem{mcgregor2013sketching}
A.~McGregor and D.~Stubbs.
\newblock Sketching earth-mover distance on graph metrics.
\newblock In {\em Approximation, Randomization, and Combinatorial Optimization.
  Algorithms and Techniques}, pages 274--286. Springer, 2013.

\bibitem{melnik2002similarity}
S.~Melnik, H.~Garcia-Molina, and E.~Rahm.
\newblock Similarity flooding: A versatile graph matching algorithm and its
  application to schema matching.
\newblock In {\em Data Engineering, 2002. Proceedings.}, pages 117--128. IEEE,
  2002.

\bibitem{merris1994laplacian}
R.~Merris.
\newblock Laplacian matrices of graphs: a survey.
\newblock {\em Linear algebra and its applications}, 197:143--176, 1994.

\bibitem{myers2014bursty}
S.~A. Myers and J.~Leskovec.
\newblock {The bursty dynamics of the Twitter information network}.
\newblock In {\em Proceedings of the 23rd international conference on World
  wide web}, pages 913--924. ACM, 2014.

\bibitem{emd-pelwer}
O.~Pele and M.~Werman.
\newblock {A linear time histogram metric for improved SIFT matching}.
\newblock In {\em Computer Vision--ECCV 2008}, pages 495--508. Springer, 2008.

\bibitem{emd-rubner}
Y.~Rubner, C.~Tomasi, and L.~J. Guibas.
\newblock {The Earth Mover's Distance as a metric for image retrieval}.
\newblock {\em International Journal of Computer Vision}, 40(2):99--121, 2000.

\bibitem{tang2013earth}
Y.~Tang, U.~Leong~Hou, Y.~Cai, N.~Mamoulis, and R.~Cheng.
\newblock {Earth Mover's Distance based similarity search at scale}.
\newblock {\em Proceedings of the VLDB Endowment}, 7(4):313--324, 2013.

\bibitem{wilson2005pattern}
R.~C. Wilson, E.~R. Hancock, and B.~Luo.
\newblock Pattern vectors from algebraic graph theory.
\newblock {\em Pattern Analysis and Machine Intelligence, IEEE Transactions
  on}, 27(7):1112--1124, 2005.

\bibitem{yildiz2011discrete}
E.~Yildiz, D.~Acemoglu, A.~E. Ozdaglar, A.~Saberi, and A.~Scaglione.
\newblock Discrete opinion dynamics with stubborn agents.
\newblock {\em Available at SSRN 1744113}, 2011.

\bibitem{zhu2005study}
P.~Zhu and R.~C. Wilson.
\newblock A study of graph spectra for comparing graphs.
\newblock In {\em BMVC}, 2005.

\end{thebibliography}
}

\end{document}